\documentclass[11pt]{article}
\usepackage{color}
\usepackage{amssymb}
\usepackage{enumerate}
\usepackage{amsmath,amsthm}
\usepackage{epsfig}
\usepackage{setspace}
\usepackage{verbatim}
\usepackage[square, sort, numbers]{natbib}
\usepackage{epstopdf}
\def\Q{{\mathbb Q}}        
\def\R{{\mathbb R}}        
\def\P{{\mathbb P}}        
\def\E{{\mathbb E}}        
\def\1{{\mathbf 1}}        
\def\F{{\mathcal F}}        
\def\G{{\mathcal G}}        
\DeclareMathOperator*{\esssup}{ess\,sup}
\DeclareMathOperator*{\essinf}{ess\,inf}

\newtheorem{theorem}{Theorem}[section]

\newtheorem{lemma}[theorem]{Lemma}

\newtheorem{proposition}[theorem]{Proposition}

\theoremstyle{remark}
\newtheorem{remark}[theorem]{Remark}
\newtheorem{example}[theorem]{Example}
\numberwithin{equation}{section}
\numberwithin{theorem}{section}

\begin{document}
\title{Risk Premia and Optimal Liquidation of Credit Derivatives\thanks{First draft: September 26, 2011. Revised: June 30, 2012, and September 25, 2012. }}

\author{Tim Leung\thanks{IEOR Department, Columbia University, New York, NY 10027; email:\,\mbox{leung@\,ieor.columbia.edu}. Work partially  supported by NSF grant DMS-0908295.}\\Columbia University  \and
Peng Liu\thanks{Applied Mathematics and Statistics Department, Johns
Hopkins University, Baltimore, MD 21218; email:\,\mbox{pliu19@jhu.edu}.
}\\Johns Hopkins University}
\date{\small{September 25, 2012}}
 \maketitle

\begin{abstract}This paper studies the optimal timing to liquidate credit derivatives in a general intensity-based credit risk model  under stochastic interest rate. We incorporate the  potential price discrepancy between the market and investors, which is characterized by risk-neutral valuation under different default risk premia  specifications.  We quantify the value of optimally timing to sell through the concept of  {delayed liquidation premium}, and analyze the associated probabilistic representation and variational inequality.  We illustrate the optimal liquidation policy for both single-named and multi-named credit derivatives. Our model is extended  to study   the sequential buying and selling problem with and without short-sale constraint.
\end{abstract}

\begin{small}
 {\textbf{Keywords:}~optimal liquidation, credit derivatives, price discrepancy, default risk premium, event risk premium}\\

 {\textbf{JEL Classification:}~G12, G13, C68}
\end{small}

\newpage
\tableofcontents
\newpage

\section{Introduction}\label{sect-intro}
In credit derivatives trading, one  important question is how the market compensates investors for bearing credit risk. A number of studies  \cite{Azizpour2008,Berndt2005,Driessen2005,Jarrow2005} have examined analytically and empirically the structure  of  default risk premia inferred from the market prices of  corporate bonds, credit default swaps, and multi-name credit derivatives. A major risk premium component  is the  \emph{mark-to-market risk premium} which accounts for  the fluctuations in default risk.  Under reduced-form models of  credit risk \cite{Duffie1999,Jarrow1995, Lando1998}, this is connected with  a drift  change of the state variable diffusion driving  the default intensity.  In addition, there is the \emph{event risk premium} (or jump-to-default  risk premium) that compensates for  the uncertain timing of the default event, and is measured by the ratio of the risk-neutral intensity to the historical intensity (see \cite{Azizpour2008,Jarrow2005}).

From standard no-arbitrage pricing theory, risk premia specification is  inherently tied  to the selection of risk-neutral  pricing measures.  A typical buy-side investor (e.g. hedge fund manager or proprietary trader) would identify trading opportunities by looking for mispriced contracts in the market. This can be interpreted as selecting a pricing measure to reflect her view on credit risk evolution  and the required risk premia. As a result, the investor's pricing measure may differ from that represented by the prevailing market prices. In a related study,  Leung and Ludkovski \cite{LeungLudkovski2}  showed that  such a price discrepancy  would also arise from pricing under marginal utility.

Price discrepancy  is also important for investors with  credit-sensitive  positions who may need to control risk exposure through liquidation. The central issue lies in the timing of liquidation as investors have the option to sell at the current market price or wait for a later  opportunity. The optimal strategy, as we will study,  depends on  the sources of risks, risk premia, as well as derivative payoffs.

This  paper tackles the  optimal liquidation problem on two fronts. First, we provide a general mathematical framework for  price discrepancy between the market and investors under an intensity-based credit risk model.  Second, we derive and analyze  the optimal stopping problem corresponding to  the liquidation of  credit derivatives under price discrepancy.

 In order to measure the benefit of optimally timing to sell as opposed to immediate liquidation, we employ the concept of \emph{delayed liquidation premium}.  It turns out to be a very useful tool for analyzing the optimal stopping problem. The intuition is that the investor should wait as long as the delayed liquidation premium is strictly positive. Applying  martingale arguments, we deduce the scenarios where  immediate  or  delayed liquidation is optimal (see Theorem \ref{thm_main}).  Moreover, through its  probabilistic representation, the delayed liquidation premium  reveals the roles of risk premia in the liquidation timing. Under a Markovian credit risk model, the optimal timing is characterized by a liquidation boundary solved from a variational inequality.  For numerical illustration, we provide a series of examples where  the default intensity and interest rate follow Cox-Ingersoll-Ross (CIR) or Ornstein-Uhlenbeck (OU) processes.

Our study also considers the connection between different risk-neutral pricing measures (or equivalent martingale measures) in incomplete markets. Well-known examples of candidate pricing measures  that are consistent with the no-arbitrage  principle include the minimal martingale measure \cite[]{FollmerSchweizer1990}, the minimal entropy martingale measure
\citep{Frittelli00, Fujiwara2003}, and the $q$-optimal martingale measure
\cite[]{Hobson2004,HHHS}. The investor's selection of various pricing measures may also be interpreted via marginal utility indifference valuation (see, among others, \cite{Davis97,LeungLudkovski2,LSZ2}).

 For many parametric credit risk models, the market pricing measures and risk premia can be calibrated given sufficient market data of  credit derivatives.  For instance, Berndt et al. \cite{Berndt2005} estimated default risk premia from credit default swap (CDS) rates and Moody’s KMV expected default frequency (EDF) measure.  For CDO tranche spreads,  Cont and Minca \cite{Cont2011} constructed a pricing measure and default intensity based on entropy minimization.  In this paper, we   focus on investigating  the impact of pricing measure on the investor's liquidation timing for various credit derivatives, including defaultable bonds, CDSs, as well as, multi-name credit derivatives.

In  recent literature, a number of models have been proposed  to incorporate the idea of mispricing into optimal investment.
Cornell et al. \cite{Cornell2007} studied portfolio optimization based on  perceived mispricing from the investor's strong belief in the stock price distribution.
Ekstr\"{o}m et al. \cite{Ekstrom2010} investigated the optimal liquidation of a call spread
when the investor's belief on the volatility differs from the implied volatility.  On the other hand, the
problem of optimal  stock liquidation involving price impacts has been studied in
\cite{Almgren2003,Rogers2010,Schied2009}, among others.

Our work is closest in spirit to \cite{LeungLudkovski2011} where the delayed purchase premium concept was used  to analyze  the optimal timing to  {purchase} equity European and American options under  a stochastic volatility model and a defaultable stock model.   In contrast,  the current paper  addresses the optimal timing to  {liquidate} various  credit derivatives.  In particular, we adopt a multi-factor  intensity-based  default  risk model for single-name credit derivatives, and a self-exciting top-down model for a credit default index swap. As a natural extension, we also investigate  the optimal timing to  buy and sell   a credit  derivative, with or without short-sale constraint, and provide numerical illustration of the   the optimal buy-and-sell strategy.

The rest of the paper is organized as follows. In Section 2, we present the
mathematical model for price discrepancy and  formulate the optimal
liquidation problem under a general intensity-based credit risk model. In
Section 3, we study the problem  within  a Markovian  market  and
characterize the optimal liquidation strategy for a general defaultable claim. In
Section 4, we apply our analysis  to a number of single-name credit derivatives, e.g. defaultable bonds  and credit
default swaps (CDS).  In Section 5, we discuss the optimal liquidation of credit default index swap. In Section 6, we examine the  optimal buy-and-sell strategy for defaultable claims.  Section 7 concludes the paper and suggests directions for future research.

\section{Problem Formulation}\label{sect-overview}

This section provides  the mathematical formulation of price discrepancy and  the optimal liquidation of credit derivatives under an  intensity-based credit risk model. We fix a probability space $(\Omega, \G, \P)$, where $\P$ is the historical measure, and denote $T$ as the maturity of derivatives in question. There is a stochastic risk-free interest rate process   $(r_t)_{0\leq t\leq T}$. The default  arrival  is described by the first jump of a doubly-stochastic Poisson process. Precisely, assuming a default intensity process $(\hat{\lambda}_t)_{0\leq t\leq T}$, we define  the  default time $\tau_d$  by
\begin{align}
\tau_d=\inf\{t\ge 0: \int_0^t \! \hat{\lambda}_s ds\ >E\},  \qquad \text{where}\  E\sim Exp\ (1)\, \,\text{and }\, E\perp\hat{\lambda}, r.
\end{align}The associated default counting process is  $N_t=\1_{\{t\geq\tau_d\}}$.  The filtration  $\mathbb{F}=(\F_t)_{0\leq t\leq T}$ is generated by $r$ and $\hat{\lambda}$. The full filtration $\mathbb{G}=(\G_t)_{0\leq t\leq T}$  is defined by $\mathcal{G}_t=\mathcal{F}_t\vee\mathcal{F}^N_t$ where $(\mathcal{F}^N_t)_{0\le t\le T}$ is generated by $N$ (see e.g. \citep[Chap. 5]{Schonbucher2003}).

\subsection{Price Discrepancy}
By  standard no-arbitrage pricing theory, the market price of a defaultable claim, denoted by $(P_t)_{0\leq t\leq T}$, is computed from a conditional expectation of discounted payoff under the market risk-neutral (or equivalent martingale) pricing measure $\Q\sim\P$. In many parametric credit risk models, the market pricing measure $\Q$ is related to the historical measure $\P$ via the default risk premia (see Section \ref{sect-premia}  below). We assume the standard hypothesis (H) that every $\mathbb{F}$-local martingale is a $\mathbb{G}$-local martingale holds under $\Q$ (see \cite[\S 8.3]{Bielecki2002}).

We can describe a general  defaultable claim by the quadruple  $(Y,
 A, R, \tau_d)$, where $Y \in \F_T$ is the terminal payoff if the defaultable claim survives at
 $T$, $(A_t)_{0\leq t\leq T}$ is a $\mathbb{F}$-adapted continuous
 process of finite variation with $A_0=0$ representing the promised
 dividends until maturity or default, and $(R_t)_{0\leq t\leq T}$ is a
 $\mathbb{F}$-predictable process representing the recovery payoff paid at
 default.  Similar notations are used by Bielecki et al. \cite{Bielecki2008} where  the following
 integrability conditions are assumed:
\begin{align}
&\E^{\Q}\big\{\big|e^{-\int_0^T  r_v  dv}Y\big|\big\}<\infty, \quad \E^{\Q}\big\{\big|\int_{(0,T]}e^{-\int_0^u  r_v  dv}(1-N_u)dA_u\big|\big\}<\infty, \quad \text{and}\nonumber\\
&\E^{\Q}\big\{\big|e^{-\int_0^{\tau_d\wedge T}  r_v  dv}R_{\tau_d\wedge T}\big|\big\}<\infty.\label{integrability}
\end{align}

For a defaultable claim  $(Y, A, R, \tau_d)$, the associated cash flow process
$(D_t)_{0\leq t\leq T}$  is defined by
\begin{align}
D_t:=Y\1_{\{\tau_d>T\}}\1_{\{t\geq T\}}+\int_{(0,t\wedge T]}  (1-N_u)dA_u+\int_{(0,t\wedge T]}  R_u dN_u. \label{Dt}
\end{align}
Then, the (cumulative) market price process $(P_t)_{0\leq t\leq T}$ is given by the conditional expectation under the market pricing measure $\Q$\,:
\begin{align}
P_t:=\E^{\Q}\big\{\int_{(0,T]} e^{-\int_t^u  r_v  dv}dD_u|\G_t\big\}. \label{P_price}
\end{align}One simple example is the zero-coupon zero-recovery defaultable bond $(1, 0, 0, \tau_d)$, whose  market price  is simply  $P_t =\E^\Q\big\{e^{-\int_t^T  r_v dv}\1_{\{\tau_d>T\}}|\G_t\big\}$.

When a perfect replication is unavailable,  the market is incomplete and there  exist different
risk-neutral pricing measures that give different no-arbitrage prices for the same defaultable claim. Mathematically, this amounts to assigning a
different risk-neutral pricing measure $\tilde{\Q}\sim \Q$. The
investor's reference price process $(\tilde{P}_t)_{0\leq t\leq T}$ is given by the conditional
expectation under investor's risk-neutral pricing measure $\tilde{\Q}$\,:
\begin{align}
\tilde{P}_t:=\E^{\tilde{\Q}}\big\{\int_{(0,T]} e^{-\int_t^u  r_v  dv}dD_u|\G_t\big\}, \label{investor Pr}
\end{align}
whose discounted price process $(e^{-\int_0^{t} \! r_v  dv}\tilde{P}_t)_{0\le t\leq T}$ is  a $(\tilde{\Q}, \mathbb{G})$-martingale.   We assume that the standard hypothesis (H) also holds under $\tilde{\Q}$.

\subsection{Optimal Stopping \& Delayed Liquidation Premium}\label{sect-delayed}
A defaultable claim holder can sell  her position at the prevailing  market price.  If she completely agrees with the market price, then she will be indifferent to sell at any time. Under price discrepancy, however,  there is a timing option embedded in the optimal liquidation problem. Precisely, in order to maximize the expected spread between the investor's price and the market price, the holder solves the optimal stopping problem:
\begin{align}
J_t:=\esssup_{\tau \in \mathcal{T}_{t,T}}~\E^{\tilde{\Q}}\big\{e^{-\int_t^{\tau}  r_v  dv}(P_\tau - \tilde{P}_\tau)|\G_t\big\}, \quad 0\leq t \leq T, \label{J_eq}
\end{align}
where $\mathcal{T}_{t,T}$ is the set of $\mathbb{G}$-stopping times  taking values in $[t,T]$.
Using repeated conditioning, we decompose \eqref{J_eq} to $J_t = V_t - \tilde{P}_t$, where
\begin{align}
V_t:=\esssup_{\tau \in \mathcal{T}_{t,T}}~\E^{\tilde{\Q}}\big\{e^{-\int_t^{\tau}  r_v  dv}P_\tau|\G_t\big\}\label{V_eq}.
\end{align}Hence, maximizing the price spread in \eqref{J_eq} is equivalent to maximizing the expected discounted future market value  $P_\tau$ under the investor's measure $\tilde{\Q}$ in \eqref{V_eq}.

The selection of the risk-neutral  pricing measure $\tilde{\Q}$ can be based on the investor's hedging criterion or  risk preferences.  For instance,
dynamic hedging under a quadratic criterion amounts to pricing under the well-known minimal martingale measure  developed by F\"{o}llmer and Schweizer \cite{FollmerSchweizer1990}.   On the other hand, different risk-neutral pricing measures may also arise from  marginal utility indifference pricing. In the cases of  exponential and power utilities, this pricing mechanism will lead the investor to select the minimal entropy martingale measure (MEMM) (see \cite{LeungLudkovski2}) and the $q$-optimal martingale measure (see \cite{HHHS}).

\begin{lemma}  For  $0\le t\le T$, we have $V_t \ge P_t \vee \tilde{P}_t$. Also, $V_{\tau_d}=\tilde{P}_{\tau_d}=P_{\tau_d}$ at default.
\end{lemma}
\begin{proof}
Since  $\tau=t$ and $\tau=T$ are candidate liquidation times, we conclude from \eqref{V_eq} that $V_t \ge P_t \vee \tilde{P}_t$. Also,
 we observe from \eqref{Dt} that $P_t= \int_{(0,\tau_d]} e^{-\int_t^u  r_v  dv}dD_u=\tilde{P}_t$ for $t\geq \tau_d \wedge T$.  This implies that \begin{align}
V_{\tau_d}=\esssup_{\tau \in \mathcal{T}_{\tau_d,T}}~\E^{\tilde{\Q}}\big\{e^{-\int_{\tau_d}^{\tau}  r_v  dv}P_\tau|\G_{\tau_d}\big\}=\esssup_{\tau \in \mathcal{T}_{\tau_d,T}}~\E^{\tilde{\Q}}\big\{e^{-\int_{\tau_d}^{\tau}  r_v  dv}\tilde{P}_\tau|\G_{\tau_d}\big\}=\tilde{P}_{\tau_d}=P_{\tau_d} \label{VP_equal}.
\end{align}
\end{proof}
The last equation  means that  price discrepancy vanishes when the default event is observed or when the contract expires. This is also realistic since the market will  no longer be liquid afterward.

If the defaultable claim is   underpriced by the market at all times, that is, $P_t \le \tilde{P}_t$,  $\forall t\le T$,  then we infer from (\ref{J_eq}) that $J_t=0$. This can be achieved at $\tau^*=T$  since price discrepancy must vanish at maturity, i.e. $P_T= \tilde{P}_T$.     In turn, this  implies that \[V_t =\E^{\tilde{\Q}}\big\{e^{-\int_t^{T}  r_v  dv}P_T|\G_t\big\}=\E^{\tilde{\Q}}\big\{e^{-\int_t^{T}  r_v  dv}\tilde{P}_T|\G_t\big\}=\tilde{P}_t.\]
In this case, there is no benefit to liquidate before maturity $T$.

According to \eqref{V_eq}, the  optimal liquidation timing directly depends on the investor's pricing measure $\tilde{\Q}$ as well as the market pricing measure $\Q$ (via the market price $P$).  Specifically, we observe that the discounted market price $(e^{-\int_0^{t} \! r_v  dv}P_t)_{0\le t\leq T}$ is a $(\Q, \mathbb{G})$-martingale, but generally not a  $(\tilde{\Q}, \mathbb{G})$-martingale. If the discounted market price is a $(\tilde{\Q}, \mathbb{G})$-\emph{supermartingale}, then it is optimal to sell the claim immediately. If the discounted market price turns out to be a $(\tilde{\Q}, \mathbb{G})$-\emph{submartingale}, then  it is optimal to delay the liquidation until maturity $T$. Besides these two scenarios, the optimal liquidation strategy may be non-trivial.

To quantify the value of optimally waiting to sell, we define the \emph{delayed liquidation premium}:
\begin{align}
L_t:=V_t-P_t\ge 0.\label{def_L}
\end{align}
It is often more intuitive to study the optimal liquidation timing in terms of the premium  $L$. Indeed, standard optimal stopping theory \citep[Appendix D]{KaratzasShreve01} suggests that  the optimal stopping time
$\tau^{\ast}$ for \eqref{V_eq} is  the first time the process $V$ reaches the reward $P$, namely,
\begin{align}
\tau^{\ast}&=\inf\{t\leq u \leq T: V_u=P_u\}=\inf\{t\leq u \leq T: L_u=0\}.\label{tau}
\end{align}The last equation, which follows directly from definition \eqref{def_L}, implies that the investor will liquidate as soon as the delayed liquidation premium vanishes. Moreover, we   observe from \eqref{VP_equal} and \eqref{tau} that $\tau^* \leq \tau_d$.

\section{Optimal Liquidation under Markovian Credit Risk Models}\label{sect-credit derivative}
We proceed to  analyze the optimal liquidation problem under a general class of Markovian credit risk models. The description of various pricing measures will involve the mark-to-market risk premium and event risk premium, which are crucial  in the characterization of the optimal liquidation strategy (see Theorem \ref{thm_main}).

\subsection{Pricing Measures and Default Risk Premia} \label{sect-premia}
We consider a  $n$-dimensional Markovian state vector process ${\bf X}$   that drives   the interest rate $r_t=r(t,{\bf X}_t)$ and  default intensity
$\hat{\lambda}_t=\hat{\lambda}(t,{\bf X}_t)$ for some positive measurable
functions $r(\cdot, \cdot)$ and $\hat{\lambda}(\cdot, \cdot)$.  Denote by $\mathbb{F}$ the filtration generated by ${\bf X}$.  We also assume   a Markovian payoff
structure for the defaultable claim $(Y, A, R, \tau_d)$  with $Y=Y({\bf
X}_T)$, $A_t=\int_0^tq(u,{\bf X}_u)du$,  and $R_t=R(t,{\bf X}_t)$ for some
measurable functions $Y(\cdot)$, $q(\cdot, \cdot)$, and $R(\cdot,\cdot)$ satisfying integrability conditions (\ref{integrability}).

Under the historical measure $\P$, the state vector process ${\bf X}$ satisfies the SDE
\begin{align}
d{\bf X}_t=a(t,{\bf X}_t)dt+\Sigma(t,{\bf X}_t)d{\bf W}^\P_t, \label{SDE_X_P}
\end{align}
where ${\bf W}^\P$ is a $m$-dimensional $\P$-Brownian motion, $a$ is the deterministic drift function, and $\Sigma$ is the $n$ by $m$ deterministic volatility function.  Standard Lipschitz and growth conditions \cite[\S 5.2]{KaratzasShreve91}  are assumed to  guarantee a unique solution to (\ref{SDE_X_P}).

Next, we consider the market pricing measure $\Q \sim \P$. To this end, we define  the Radon-Nikodym density process $(Z^{\Q, \P}_t)_{0\le t\leq T}$   by
\begin{equation}
Z^{\Q, \P}_t=\frac{d\Q}{d\P}\big|\mathcal{G}_t=\mathcal{E}\big(- {\boldsymbol \phi}^{\Q, \P}\!\cdot\! {\bf W}^\P\big)_t\,\mathcal{E}\big((\mu-1) M^\P\big)_t\,, \label{Z_t_P_Q}
\end{equation}
where the Dol\'eans-Dade exponentials are defined by
\begin{align}
\mathcal{E}\big(- {\boldsymbol \phi}^{\Q, \P}\!\cdot\! {\bf W}^\P\big)_t&:=\text{exp}\bigg(-\frac{1}{2}\int_0^t \! |\!| {\boldsymbol \phi}^{\Q, \P}_u |\!|^2 du-\int_0^t \! {\boldsymbol \phi}^{\Q, \P}_u \!\cdot\! d{\bf W}^\P_u\bigg),  \label{measure_W_PQ}\\
\mathcal{E}\big((\mu-1) M^\P\big)_t&:=\text{exp}\bigg(\int_0^t \! \text{log}(\mu_{u-})  dN_u-\int_0^t \!(1-N_u)(\mu_u-1)\hat{\lambda}_u du\bigg)\label{measure_M_P_Q},
\end{align}
and $M^\P_t:=N_t-\int_0^t \! (1-N_u)\hat{\lambda}_u du$ is the compensated $(\P, \mathbb{G})$-martingale associated with $N$. Here,   $({\boldsymbol \phi}_t^{\Q,\P})_{0\leq t\leq T}$ and $(\mu_t)_{0\leq t\leq T}$  are adapted processes  satisfying
$\int_0^T  |\!| {\boldsymbol \phi}^{ {\Q}, \P}_u |\!|^2 du<\infty$, $\mu \ge0$,  and  $\int_0^T \mu_u\hat{\lambda}_u du<\infty$ (see Theorem
4.8 of  \cite{Schonbucher2003}).

The process  ${\bf {\boldsymbol \phi}}^{\Q,\P}$ is commonly referred to as the \emph{mark-to-market risk premium} (see \citep{Azizpour2008}), which is assumed herein to be Markovian of the form ${\boldsymbol \phi}^{\Q,\P}(t, {\bf X}_t)$. The process $\mu$ is referred to as \emph{event risk premium} (see \citep{Azizpour2008,Jarrow2005}), which captures the compensation from the uncertain timing of default.
The $\Q$-default intensity, denoted by $\lambda$,  is related to $\P$-intensity via $\lambda_t=\mu_t\hat{\lambda}_t$.  Here, we also assume $\mu$ to be Markovian of the form $\mu(t,{\bf X}_t) ={\lambda(t,{\bf X}_t) }/{\hat{\lambda}(t,{\bf X}_t) }$.

By multi-dimensional Girsanov Theorem, it follows that ${\bf W}^{\Q}_t:={\bf W}^\P_t+\int_0^t{\boldsymbol \phi}_u^{\Q,\P}du$ is a $m$-dimensional $\Q$-Brownian motion, and $M^{\Q}_t:=N_t-\int_0^t (1-N_u)\mu_u\hat{\lambda}_u du$ is a $(\Q, \mathbb{G})$-martingale. Consequently, the $\Q$-dynamics of  ${\bf X}$ are given by
\begin{align}
d{\bf X}_t=b(t,{\bf X}_t)dt+\Sigma(t,{\bf X}_t)d{\bf W}^\Q_t,\label{SDE_X_Q}
\end{align}
where $b(t,{\bf X}_t):=a(t,{\bf X}_t)-\Sigma(t,{\bf X}_t){\boldsymbol
\phi}^{\Q,\P}(t,{\bf X}_t)$.

Similarly, the investor's pricing measure $\tilde{\Q}$ is related to the historical measure $\P$ through the investor's Markovian risk premium functions  ${\boldsymbol \phi}^{\tilde{\Q},\P}(t, {\bf x})$ and $\tilde{\mu}(t, {\bf x})$. Precisely, the measure $\tilde{\Q}$ is defined by the density process $
Z^{\tilde{\Q}, \P}_t=\mathcal{E}\big(- {\boldsymbol \phi}^{\tilde{\Q}, \P}\!\cdot\! {\bf W}^\P\big)_t\,\mathcal{E}\big((\tilde{\mu}-1) M^\P\big)_t$.
By a change of measure, the drift of ${\bf X}$ under $\tilde{\Q}$ is modified to $\tilde{b}(t,{\bf X}_t) :=a(t,{\bf X}_t)-\Sigma(t,{\bf X}_t){\boldsymbol \phi}^{\tilde{\Q},\P}(t,{\bf X}_t)$.

 Then, the EMMs  $\Q$ and $\tilde{\Q}$ are related by
the  Radon-Nikodym derivative:
\begin{equation}
Z^{\tilde{\Q}, \Q}_t=\frac{d\tilde{\Q}}{d\Q}\big|\mathcal{G}_t=\mathcal{E}\big(- {\boldsymbol \phi}^{\tilde{\Q}, \Q}\!\cdot\! {\bf W}^\Q\big)_t\,\mathcal{E}\big((\frac{\tilde{\mu}}{\mu}-1) M^\Q\big)_t\,, \label{Z_t}
\end{equation}
where the Dol\'eans-Dade exponentials are defined by
\begin{align}
\mathcal{E}\big(- {\boldsymbol \phi}^{\tilde{\Q}, \Q}\!\cdot\! {\bf W}^\Q\big)_t&:=\text{exp}\bigg(-\frac{1}{2}\int_0^t \! |\!| {\boldsymbol \phi}^{\tilde{\Q}, \Q}_u |\!|^2 du-\int_0^t \! {\boldsymbol \phi}^{\tilde{\Q}, \Q}_u \!\cdot\! d{\bf W}^\Q_u\bigg),  \label{measure_W}\\
\mathcal{E}\big((\frac{\tilde{\mu}}{\mu}-1) M^\Q\big)_t&:=\text{exp}\bigg(\int_0^t \! \text{log}(\frac{\tilde{\mu}_{u-}}{\mu_{u-}})  dN_u-\int_0^t \!(1-N_u)(\frac{\tilde{\mu}_u}{\mu_u}-1) {\lambda}_u du\bigg)\label{measure_M}.
\end{align}
We observe that ${\boldsymbol \phi}_t^{\tilde{\Q}, \Q} = {\boldsymbol \phi}_t^{\tilde{\Q},\P}-{\boldsymbol \phi}_t^{\Q,\P}$ from the decomposition:
\begin{align}
{\boldsymbol \phi}_t^{\tilde{\Q}, \Q}dt=d{\bf W}^{\tilde{\Q}}_t-d{\bf W}^{\Q}_t=(d{\bf W}^{\tilde{\Q}}_t-d{\bf W}^{\P}_t)-(d{\bf W}^{\Q}_t-d{\bf W}^{\P}_t)=({\boldsymbol \phi}_t^{\tilde{\Q},\P}-{\boldsymbol \phi}_t^{\Q,\P})dt.\label{decomp}
\end{align} Therefore, we can interpret ${\boldsymbol
\phi}^{\tilde{\Q}, \Q}$  as the incremental mark-to-market risk premium assigned by the investor relative to the market. On the other hand, the  discrepancy in event risk premia is accounted for in the second  Dol\'eans-Dade exponential \eqref{measure_M}.

\begin{example}\label{sect-OU}\emph{The OU Model.}  Suppose $ (r,\hat{\lambda}) ={\bf X}$, following  the OU dynamics:
\begin{align}
\begin{pmatrix} dr_t \\ d\hat{\lambda}_t\end{pmatrix}=\begin{pmatrix} \hat{\kappa}_r(\hat{\theta}_r-r_t) \\ \hat{\kappa}_{\lambda}(\hat{\theta}_{\lambda}-\hat{\lambda}_t)\end{pmatrix}dt+ \begin{pmatrix} \sigma_r & 0  \\\sigma_{\lambda}\rho & \sigma_{\lambda}\sqrt{1-\rho^2}\end{pmatrix} \begin{pmatrix} dW_t^{1,\P} \\ dW_t^{2,\P}\end{pmatrix},\label{P_state}
\end{align}
with constant parameters $\hat{\kappa}_r, \hat{\theta}_r, \hat{\kappa}_{\lambda}, \hat{\theta}_{\lambda}\geq 0$. Here,  $\hat{\kappa}_r$, $\hat{\kappa}_{\lambda}$ parameterize the speed of mean reversion, and $\hat{\theta}_r$, $\hat{\theta}_{\lambda}$ represent the long-term means (see \cite[\S 7.1.1]{Schonbucher2003}). Assuming a constant event risk premium $\mu$ by the market, the $\Q$-intensity is specified by $\lambda_t=\mu\hat{\lambda}_t$ and the pair $(r, \lambda)$ satisfies SDEs:
\begin{align}
\begin{pmatrix} dr_t \\ d\lambda_t\end{pmatrix}=\begin{pmatrix} \kappa_r( \theta_r-r_t) \\ \kappa_{\lambda}(\mu\theta_{\lambda}-\lambda_t)\end{pmatrix}dt+ \begin{pmatrix} \sigma_r & 0  \\\mu\sigma_{\lambda}\rho & \mu\sigma_{\lambda}\sqrt{1-\rho^2}\end{pmatrix} \begin{pmatrix} dW_t^{1,\Q} \\ dW_t^{2,\Q}\end{pmatrix}, \label{Q_state}
\end{align}
with constants $\kappa_r, \theta_r, \kappa_{\lambda}, \theta_{\lambda}\geq 0$. Under the investor's   measure $\tilde{\Q}$, the SDEs for $r_t$ and $\tilde{\lambda}_t= \tilde{\mu}\hat{\lambda}_t$ are of the same form with parameters $\tilde{\kappa}_r, \tilde{\theta}_r, \tilde{\kappa}_{\lambda}, \tilde{\theta}_{\lambda}$ and $\tilde{\mu}$,  and ${\bf W}^\Q$ is replaced by ${\bf W}^{\tilde{\Q}}$.

Direct computation yields the relative mark-to-market  risk premium:
\begin{align}
{\boldsymbol \phi}_t^{\tilde{\Q}, \Q}&= \begin{pmatrix} \frac{\kappa_r(\theta_r-r_t)-\tilde{\kappa}_r(\tilde{\theta}_r-r_t)}{\sigma_r} \nonumber\\ \frac{1}{\sqrt{1-\rho^2}}\frac{\kappa_{\lambda}(\theta_{\lambda}-\hat{\lambda}_t)-\tilde{\kappa}_{\lambda}(\tilde{\theta}_{\lambda}-\hat{\lambda}_t)}{\sigma_\lambda}-\frac{\rho}{\sqrt{1-\rho^2}}\frac{\kappa_r(\theta_r-r_t)-\tilde{\kappa}_r(\tilde{\theta}_r-r_t)}{\sigma_r}\end{pmatrix}.
\end{align}
The upper  term is the incremental risk premium for  the interest rate while the bottom term reflects the discrepancy  in the default risk premia (see \eqref{decomp}).
\end{example}

\begin{example}\label{sect-CIR}\emph{The CIR Model}. Let ${\bf X}=(X^1,\ldots,X^n)^T$ follow the multifactor CIR model \cite[\S 7.2]{Schonbucher2003}:
\begin{equation}
dX_t^i=\hat{\kappa}_i(\hat{\theta}_i-X_t^i)dt+\sigma_i \sqrt{X_t^i}\, dW_t^{i,\P},\label{P_state_CIR}
\end{equation}
where $W^{i,\P}$ are mutually independent $\P$-Brownian motions and $\hat{\kappa}_i$, $\hat{\theta}_i$, $\sigma_i \geq 0$, $i=1, \ldots, n$ satisfy Feller condition
$2\hat{\kappa}_i\hat{\theta}_i>\sigma_i^2$. The
interest rate $r$ and historical default intensity $\hat{\lambda}$ are
non-negative linear combinations of     $X^i$ with constant
weights $ w^r_i, w^{\lambda}_i \ge 0$, namely, $r_t=\displaystyle\sum_{i=1}^n w^r_iX_t^i$ and $\hat{\lambda}_t=\displaystyle\sum_{i=1}^n w^\lambda_iX_t^i.$
Under   measure $\Q$, $X^i$ satisfies the   SDE:
\begin{equation}
dX_t^i=\kappa_i(\theta_i-X_t^i)dt+\sigma_i \sqrt{X_t^i}\, dW_t^{i,\Q}, \label{Q_state_CIR}
\end{equation}
with  new mean reversion speed $\kappa_i$ and  long-run mean ${\theta_i}$.

Under the investor's measure $\tilde{\Q}$, the SDE for the state vector is of the same form with new parameters $\tilde{\kappa}_i$, ${\tilde{\theta}_i}$.  The associated relative mark-to-market risk premium has following structure:
\[\phi_{i,t}^{\tilde{\Q},\Q}=\frac{\kappa_i(\theta_i-X_t^i)-\tilde{\kappa}_i(\tilde{\theta}_i-X_t^i)}{\sigma_i \sqrt{X_t^i}}.\]
The event risk premia $(\mu, \tilde{\mu})$ are assigned via  $\lambda_t=\mu \hat{\lambda}_t$ under $\Q$ and $\tilde{\lambda}_t=\tilde{\mu}\hat{\lambda}_t$ under $\tilde{\Q}$ respectively.
\end{example}

\begin{remark}
The current framework can be readily generalized to the situation where the investor needs to assume an alternative  historical measure $\tilde{\P}$. The  resulting risk premium ${\boldsymbol \phi}^{\tilde{\Q}, \Q}$ will have a third decomposition component ${\boldsymbol \phi}^{\tilde{\P}, \P}$, reflecting the difference in historical dynamics.
\end{remark}

For any defaultable claim $(Y, A, R, \tau_d)$, the  ex-dividend pre-default market price is given by (see \citep[Prop. 8.2.1]{Bielecki2002})
\begin{align}
C(t,{\bf X}_t)=\E^{\Q}\big\{e^{-\int_t^T (r_v+\lambda_v)dv}Y({\bf X}_T)+\int_t^T e^{-\int_t^u (r_v+\lambda_v)dv}\big(\lambda_u R(u,{\bf X}_u)+q(u,{\bf X}_u)\big)du|\F_t\big\}.\label{pre-default def}
\end{align}
The associated cumulative price is related to the pre-default price via
\begin{align}
P_t=(1-N_t)C(t,{\bf X}_t)+\int_0^t(1-N_u)q(u,{\bf X}_u)e^{\int_u^t r_vdv}du+\int_{(0,t]} R(u,{\bf X}_u)e^{\int_u^t r_vdv}dN_u.  \label{Pt}
\end{align}
The price  function $C(t,{\bf x})$  can be determined by solving the PDE:
\begin{align}
\left\{ \begin{array}{ll}\displaystyle
\frac{\partial C}{\partial t}(t,{\bf x})+\mathcal{L}_{b, \lambda} C(t,{\bf x})+ \lambda(t,{\bf x}) R(t,{\bf x})+ q(t,{\bf x})=0, \quad (t,{\bf  x}) \in [0,T)\times \R^n,\\
\displaystyle C(T,{\bf x})=Y({\bf x}),\quad {\bf  x} \in \R^n,\label{PDE_C_general}
\end{array} \right.
\end{align}
where $\mathcal{L}_x$ is the operator defined by
\begin{align}\label{oper}
\mathcal{L}_{b, \lambda} f=\displaystyle\sum_{i=1}^n b_i(t,{\bf x})\frac{\partial f}{\partial x_i}+\frac{1}{2}\displaystyle\sum_{i,j=1}^n (\Sigma(t,{\bf x})\Sigma(t,{\bf x})^T)_{ij} \frac{\partial^2 f}{\partial {x_i}\partial {x_j}}-\big(r(t,{\bf x})+\lambda(t,{\bf x})\big)f.
\end{align}
The computation is similar for the investor's price  under $\tilde{\Q}$.
\subsection{Delayed Liquidation Premium and Optimal Timing}
Next, we analyze the optimal liquidation problem $V$ defined in
(\ref{V_eq})  for the general defaultable claim under the current Markovian
setting.

\begin{theorem} \label{thm_main} For a general defaultable claim  $(Y, A, R, \tau_d)$ under the Markovian credit risk model, the delayed liquidation premium admits the probabilistic representation:
\begin{align}
L_t=\1_{\{t<\tau_d\}}\esssup_{\tau \in \mathcal{T}_{t,T}}~\E^{\tilde{\Q}}\big\{\int_t^{\tau}  e^{-\int_t^u  (r_v+\tilde{\lambda}_v)  dv}G(u,{\bf X}_u) du |\F_t\big\},\label{L_formula}
\end{align}
where $G: [0,T]\times \R^n \mapsto \R$ is defined by
\begin{align}
G(t,{\bf x})&=-\big(\nabla_x C(t,{\bf x})\big)^T\Sigma(t,{\bf x}){\boldsymbol \phi}^{\tilde{\Q},\Q}(t,{\bf x})+\big(R(t,{\bf x})-C(t,{\bf x})\big)\big(\tilde{\mu}(t,{\bf x})-\mu(t,{\bf x})\big)\hat{\lambda}(t,{\bf x}). \label{G_general}
\end{align}
If $G(t,{\bf x})\geq0$ $\forall (t,{\bf x})$, then it is optimal to delay the liquidation till  maturity $T$.\\ If $G(t,{\bf x})\leq0$ $\forall (t,{\bf x})$, then  it is optimal to sell immediately.
\end{theorem}

\begin{proof} First, we look at the $\tilde{\Q}$-dynamics of discounted market  price
$(e^{-\int_t^{u} \! r_v dv}P_u)_{t\leq u\leq T}$. Applying Corollary $2.2$ of \cite{Bielecki2008}, for  $t\leq u\leq T$,
\begin{align}
d(e^{-\int_t^u \! r_v dv}P_u)&=e^{-\int_t^u \! r_v dv}[(R_u-C_u)dM^\Q_u+(1-N_u)(\nabla_x C_u)^T\Sigma_u d{\bf W}^\Q_u]\label{P_SDE}\\
&=e^{-\int_t^u \! r_v  dv}\big((1-N_u)G(u,{\bf X}_u)du+(1-N_u)(\nabla_x C_u)^T\Sigma_u d{\bf W}^{\tilde{\Q}}_u+(R_u-C_u)dM^{\tilde{\Q}}_u\big),\notag\end{align}
where $G$ is defined in \eqref{G_general}, and  $M^{\tilde{\Q}}$ is the  compensated
$(\tilde{\Q}, \mathbb{G})$-martingale for $N$. Consequently,
\begin{align}
L_t=\esssup_{\tau \in \mathcal{T}_{t,T}}~\E^{\tilde{\Q}}\big\{\int_t^{\tau}  (1-N_u) e^{-\int_t^u  r_v dv}G(u,{\bf X}_u) du |\G_t\big\},\nonumber
\end{align}
where \eqref{L_formula} follows from the change of filtration technique \cite[\S 5.1.1]{Bielecki2002}. If $G\ge 0$, then the integrand in \eqref{L_formula} is positive a.s. and therefore the largest possible stopping time $T$ is optimal. If $G \le 0$, then $\tau^*=t$ is optimal and $L_t=0$ a.s.
\end{proof}

The drift function $G$ has two components explicitly depending on ${\boldsymbol
\phi}^{\tilde{\Q},\Q}$ and $\tilde{\mu}-\mu$.  If ${\boldsymbol
\phi}^{\tilde{\Q},\Q}(t,{\bf x})={\bf 0}$  $\forall (t,{\bf x})$, that is, the investor and market agree on the mark-to-market risk premium, then the sign of
$G$ is solely determined by the difference $\tilde{\mu}-\mu$, since recovery $R$ in general is less than the
pre-default price $C$. On the other hand, if $\mu(t,{\bf x})=\tilde{\mu}(t,{\bf x})$ $\forall (t,{\bf
x})$, then the second term  of $G$
vanishes but $G$ still depends on $\mu$ through $\nabla_xC$ in the first term.

Theorem \ref{thm_main} allows us to conclude the optimal liquidation timing when  the drift function is   of constant sign. In other cases, the optimal
liquidation policy may be non-trivial and needs to be numerically
determined. For this purpose, we  write $L_t = \1_{\{t<\tau_d\}} \hat{L}(t,{\bf X}_t)$, where $\hat{L}$ is the (Markovian) pre-default delayed liquidation premium defined by
\begin{align}
\hat{L}(t,{\bf X}_t)=\esssup_{\tau \in \mathcal{T}_{t,T}}~\E^{\tilde{\Q}}\big\{\int_t^{\tau}  e^{-\int_t^u  (r_v+\tilde{\lambda}_v)  dv}G(u,{\bf X}_u) du |\F_t\big\}.\label{L_predef}
\end{align}
We determine  $\hat{L}$ from  the variational inequality:
\begin{align}
\text{min}\bigg(-\frac{\partial \hat{L}}{\partial t}(t,{\bf x})-\mathcal{L}_{\tilde{b},\tilde{\lambda}} \hat{L}(t,{\bf x})-G(t,{\bf x}), \ \hat{L}(t,{\bf x})\bigg)=0, \quad (t,{\bf x}) \in [0,T) \times \mathbb{R}^n, \label{L_general_VI}
\end{align}
where $\mathcal{L}_{\tilde{b},\tilde{\lambda}}$ is defined in \eqref{oper}, and the  terminal condition is
$\hat{L}(T,{\bf x})=0$, for $ {\bf x} \in \mathbb{R}^n$.

The investor's
optimal timing is characterized by the sell region $\mathcal S$ and delay
region $\mathcal D$, namely,
\begin{align}
\mathcal S&=\{(t,{\bf x})\in [0,T]\times \R^n : \ \hat{L}(t,{\bf x})=0\}, \label{S_region}\\
\mathcal D&=\{(t,{\bf x})\in [0,T]\times \R^n : \ \hat{L}(t,{\bf x})>0\}.\label{D_region}
\end{align}
Also, define $\hat{\tau}^{\ast}=\inf\{t\leq u \leq T: \hat{L}_u=0\}$.  On  $\{\hat{\tau}^{\ast}\geq \tau_d\}$, liquidation occurs at $\tau_d$  since $L_{\tau_d}=0$. On $\{\hat{\tau}^{\ast}<\tau_d\}$, $\hat{\tau}^{\ast}$ is optimal since when $u<\hat{\tau}^{\ast}$, $L_u=\1_{\{u<\tau_d\}} \hat{L}_u>0$ and $L_{\hat{\tau}^{\ast}}=0$. Incorporating  the observation of $\tau_d$, the optimal stopping time is  $\tau^{\ast}=\hat{\tau}^{\ast}\wedge \tau_d$.

Hence,  given no default by time $t$ and  ${\bf X}_t = {\bf x}$,  it is optimal to wait at the current time $t$ if $\hat{L}(t,{\bf x})>0$ in view of the  delay region $\mathcal{D}$ in \eqref{D_region}. This is also intuitive as there is a strictly positive premium for delaying liquidation.  On the other hand, the sell region $\mathcal{S}$ must lie within the set ${G_-}:=\{(t,{\bf x}) \,:\,  G(t,{\bf x}) \le 0\}$. To see this, we infer from  \eqref{L_predef} that, for any given point $(t,{\bf x})$ such that  $\hat{L}(t,{\bf x})=0$, we must have  $G(t,{\bf x}) \le 0$. In turn, the delay region $\mathcal{D}$ must contain the set ${G_+}:=\{(t,{\bf x}) \,:\,  G(t,{\bf x}) > 0\}$.  From these observations,  one can obtain some insights about the sell and delay regions by inspecting $G(t,{\bf x})$, which is much easier to compute than $\hat{L}(t,{\bf x})$. We shall illustrate this numerically in Figures  1-4.

Lastly, let us consider   a special example where the stochastic factor ${\bf X}$ is absent from the model. With reference to \eqref{pre-default def}, we set a constant terminal payoff $Y$, and   deterministic recovery  $R(t)$ and coupon rate $q(t)$.
Suppose the investor and market perceive the same  deterministic interest rate $r(t)$, but possibly   different deterministic  default intensities, respectively, $\tilde{\lambda}(t) = \tilde{\mu}(t) \hat{\lambda}(t)$ and $\lambda(t)= \mu(t) \hat{\lambda}(t)$. In this case, the   price function $C$  in \eqref{G_general} will depend only on $t$ but not on ${\bf x}$, and there will be no mark-to-market risk premium. Therefore,  the first term of drift function in \eqref{G_general} will vanish. However,   the second term remains due to potential discrepancy in event risk premium, i.e. $\tilde{\mu}(t) \neq \mu(t)$. As a result, the drift function reduces to \[G(t) = (R(t) - C(t))(\tilde{\mu}(t) - \mu(t))\hat{\lambda}(t).\] Furthermore, the absence of the stochastic factor ${\bf X}$ also trivializes the filtration $\mathbb{F}$, and leads the investor to optimize over only constant times. The delayed liquidation premium admits the form: $L_t = \1_{\{t<\tau_d\}} \hat{L}(t)$, where $\hat{L}(t)$ is a deterministic function given by \begin{align}
\hat{L}(t)=\sup_{t \le \hat{t}\le T}\int_t^{\hat{t}}  e^{-\int_t^u  (r(v)+\tilde{\lambda}(v))  dv}G(u) du.\label{Ltdeter}
\end{align} As in Theorem \ref{thm_main}, if $G$ is always positive (resp. negative) over $[t,T]$, then the optimal time $\hat{t}^*=T$ (resp. $\hat{t}^*=t$). Otherwise, differentiating the integral in \eqref{Ltdeter} implies  that the deterministic candidate times also include   the roots of $G(\hat{t}) = 0$. Therefore, we select among the candidate times $t, T$ and the roots of $G$ to see which would yield the largest  integral value in  \eqref{Ltdeter}.

\section{Application to Single-Name Credit Derivatives}\label{sect-selling}
We proceed to illustrate  our analysis for a number  of credit derivatives, with an emphasis on  how risk premia discrepancy affects the optimal liquidation strategies.

\subsection{Defaultable Bonds with Zero Recovery}\label{sect-bond_zeroR}
Consider a defaultable zero-coupon zero-recovery bond with face value $1$ and maturity $T$.  By a  change of filtration \cite[\S 5.1.1]{Bielecki2002}, the market price of the zero-coupon zero-recovery bond is given by
\begin{align}
P^0_t:=\E^\Q\big\{e^{-\int_t^T \! r_v dv}\1_{\{\tau_d>T\}}|\G_t\big\}=\1_{\{t<\tau_d\}}\,\E^\Q\big\{e^{-\int_t^T \! (r_v+\lambda_v) dv}|\F_t\big\}=\1_{\{t<\tau_d\}}C^0(t,{\bf X}_t),
\end{align}
where $C^0$ denotes the market pre-default price that solves \eqref{PDE_C_general}. Under the general Markovian credit risk model in Section \ref{sect-premia}, we can apply Theorem \ref{thm_main} with the quadruple $(1, 0, 0, \tau_d)$  to obtain the corresponding drift function.

Under the OU dynamics in Section \ref{sect-OU}, the pre-default price function $C^0(t,r,\lambda)$ is given explicitly by \cite[\S 7.1.1]{Schonbucher2003}:
\begin{equation}
C^0(t,r,\lambda)=e^{A(T-t)-B(T-t)r-D(T-t)\lambda}, \label{C_OU}
\end{equation}
where
\begin{align}
B(s)&=\frac{1-e^{-\kappa_r s}}{\kappa_r},\quad D(s)=\frac{1-e^{-\kappa_\lambda s}}{\kappa_\lambda},\label{B_OU}\\
A(s)&=\int_0^s \! \big[\frac{1}{2}\sigma_r^2 B^2(z)+\rho\mu\sigma_r\sigma_\lambda B(z)D(z)+\frac{1}{2}\mu^2\sigma_\lambda^2 D^2(z)-\kappa_r\theta_r B(z)-\mu\kappa_\lambda\theta_\lambda D(z)\big]  dz.\nonumber
\end{align}
As a result, the drift function $G^0(t,r,\lambda)$ admits a separable form:
\begin{align}
G^0(t,r,\lambda)&=C^0(t,r,\lambda)\bigg(B(T-t)(\tilde{\kappa}_r-\kappa_r)r+B(T-t)(\kappa_r\theta_r-\tilde{\kappa}_r\tilde{\theta}_r)\nonumber\\
&+[ D(T-t)(\tilde{\kappa}_\lambda-\kappa_\lambda)-(\frac{\tilde{\mu}}{\mu}-1)]\lambda+ \mu D(T-t)(\kappa_\lambda\theta_\lambda-\tilde{\kappa}_\lambda\tilde{\theta}_\lambda)\bigg) \label{G_OU}.
\end{align}

We can draw several insights on the liquidation timing from this drift function. If the market and the investor agree on the speed of mean reversion for interest rate, i.e. $\kappa_r=\tilde{\kappa}_r$, then $G^0(t,r,\lambda)/C^0(t,r,\lambda)$ is  linear  in $\lambda$. Furthermore, if the slope $D(T-t)(\tilde{\kappa}_\lambda-\kappa_\lambda)-(\frac{\tilde{\mu}}{\mu}-1)$ and intercept $B(T-t)(\kappa_r\theta_r-\tilde{\kappa}_r\tilde{\theta}_r)+\mu D(T-t)(\kappa_\lambda\theta_\lambda-\tilde{\kappa}_\lambda\tilde{\theta}_\lambda)$ are of the same sign, then the optimal liquidation strategy must be trivial in view of Theorem \ref{thm_main}. In contrast, if the slope and intercept differ in signs, the optimal stopping problem may be nontrivial and the sign of the slope determines qualitative properties of optimal stopping rules.  For  instance, suppose the slope is positive. We infer that it is optimal for the holder to wait at high default  intensity where the corresponding $G^0$ and thus delayed liquidation premium are positive.  The converse holds if the slope is negative.

If the investor disagrees with market only on event risk premium, i.e. $\mu \neq \tilde{\mu}$, then the drift function is reduced to $G^0(t,r,\lambda)=-C^0(t,r,\lambda)(\frac{\tilde{\mu}}{\mu}-1)\lambda$, which is of constant sign. This implies trivial strategies. If $\mu>\tilde{\mu}$, then $G^0>0$ and it is optimal to delay the liquidation until maturity. On the other hand, if $\mu<\tilde{\mu}$, then it is optimal to sell immediately. More general specifications of the event risk premium could depend on the state vector and may lead to nontrivial optimal stopping rules. Disagreement on mean level $\theta_\lambda$ has a similar effect to that of $\mu$.

If the investor disagrees with market only on speed of mean reversion, i.e. $\kappa_\lambda \neq \tilde{\kappa}_\lambda$, then $G^0(t,r,\lambda)=C^0(t,r,\lambda)D(T-t)\big[ (\tilde{\kappa}_\lambda-\kappa_\lambda)\lambda+\mu\theta_\lambda(\kappa_\lambda-\tilde{\kappa}_\lambda)\big]$ with $D(T-t)>0$ before $T$, where the slope and intercept differ in signs. If $\kappa_\lambda<\tilde{\kappa}_\lambda$, the slope $\tilde{\kappa}_\lambda-\kappa_\lambda$ is positive and it is optimal to sell immediately at a low intensity, and thus, a high bond price. The converse holds for $\kappa_\lambda>\tilde{\kappa}_\lambda$.

 We consider a numerical example where the interest rate is constant and the market default intensity $\lambda$ is chosen as the state vector ${\bf X}$ with OU dynamics. We employ the standard implicit PSOR algorithm to solve $\hat{L}(t,\lambda)$ through its variational inequality (\ref{L_general_VI}) over a uniform finite grid  with Neumann condition applied on the intensity boundary. The market parameters are $T=1$, $\mu=2$, $\kappa_\lambda=0.2$, $\theta_\lambda=0.015$, $r=0.03$, and $\sigma=0.02$, which are based on the estimates in \cite{Driessen2005,Duffee1999}.

 From formula (\ref{C_OU}), we observe  a one-to-one correspondence between the market pre-default bond price $C^0$ and its default intensity $\lambda$ for any fixed $(t,r)$, namely,
\begin{align}
\lambda=\frac{-\text{log}(C^0)+A(T-t)-B(T-t)r}{D(T-t)}. \label{inverse}
\end{align}
Substituting (\ref{inverse}) into (\ref{S_region}) and (\ref{D_region}), we can characterize the sell region and delay region in terms of the observable pre-default market  price $C^0$.

In the left panel of Figure \ref{fig1}, we assume that the investor agrees with the market on all parameters, but has a higher speed of mean reversion $\tilde{\kappa}_\lambda> \kappa_\lambda$.   In this case, the investor tends to sell the bond at a high market price, which is consistent with our previous analysis in terms of drift function.  If the bond price starts below $0.958$ at time $0$, the optimal liquidation strategy for the investor is to hold and sell the bond as soon as the price hits the optimal boundary. If the bond price starts above $0.958$ at time $0$, the optimal liquidation strategy is to sell immediately.  In the opposite case where  $\tilde{\kappa}_\lambda <\kappa_\lambda$ (see Figure \ref{fig1}(right)), the optimal liquidation strategy is reversed -- it is optimal to sell at a lower boundary. In each cases, the sell region must lie within where $G$ is non-positive, and the straight line defined by $G=0$ can be viewed as a linear approximation of the optimal liquidation boundary.

Under the CIR dynamics in Section \ref{sect-OU}, $C^0$ admits closed-form formula \cite[\S 7.2]{Schonbucher2003}
\begin{align}
C^0(t,{\bf x})=\prod_{i=1}^n \E^\Q\big\{e^{-\int_t^T \! (w^r_i+\mu w^\lambda_i)X_{v}^i  dv}|{\bf X}_t={\bf x}\big\}=\prod_{i=1}^n A_i(T-t)e^{-B_i(T-t)x_i},\label{C_CIR}
\end{align}
where
\begin{align}
A_i(s)&=[\frac{2\Xi_i e^{(\Xi_i+\kappa_i)s/2}}{(\Xi_i+\kappa_i)(e^{\Xi_i s}-1)+2\Xi_i}]^{2\kappa_i\theta_i/\sigma_i^2},\label{A_CIR}\\
B_i(s)&=\frac{2(e^{\Xi_i s}-1)(w^r_i+\mu w^\lambda_i)}{(\Xi_i+\kappa_i)(e^{\Xi_i s}-1)+2\Xi_i}, \quad\, \text{ and }\, \quad \Xi_i=\sqrt{\kappa_i^2+2\sigma_i^2(w^r_i+\mu w^\lambda_i)}\label{Xi}.
\end{align}

\begin{figure}[ht]
\centering
\begin{tabular}{cc}
\includegraphics[scale=0.47]{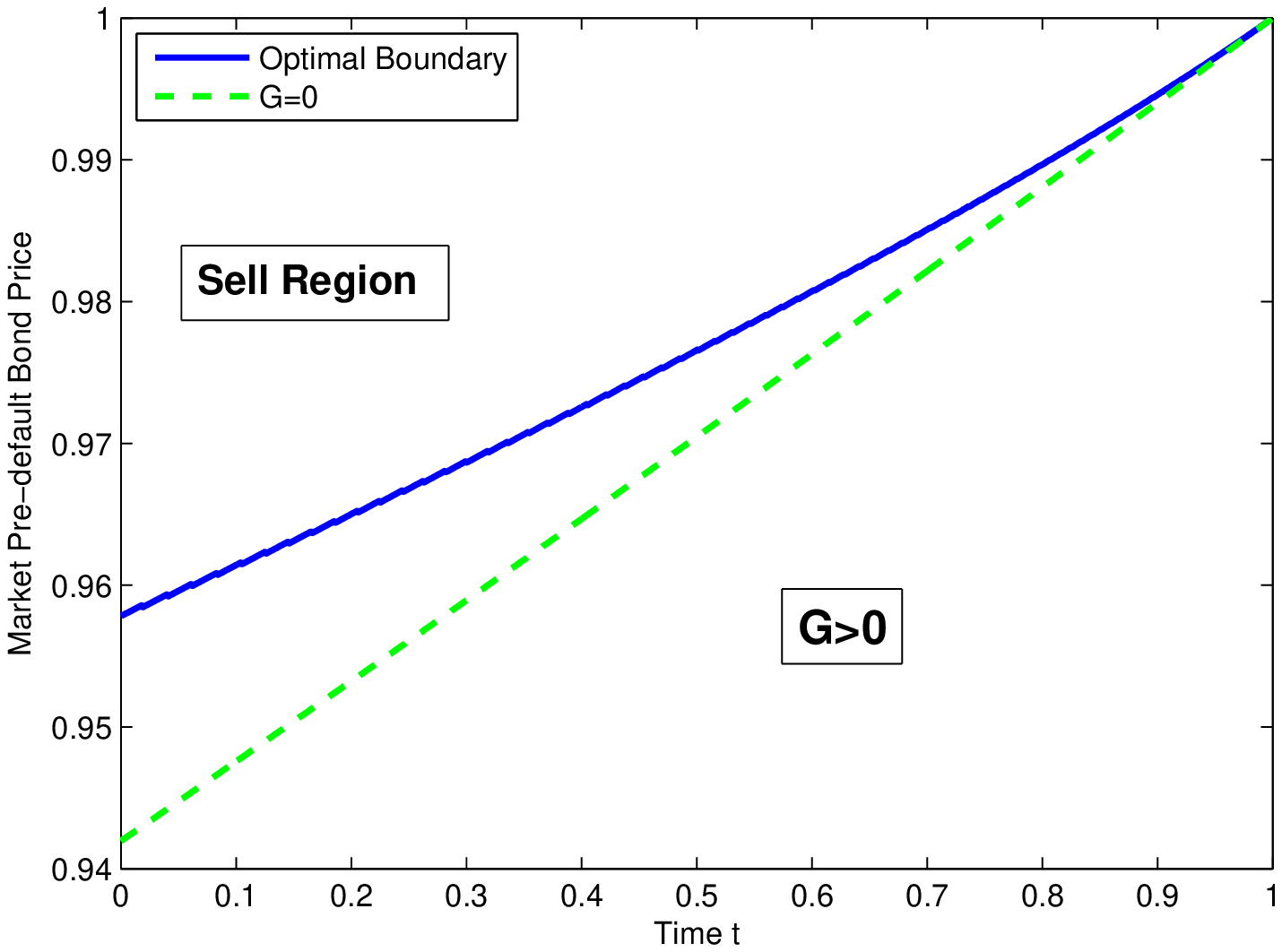} &
\includegraphics[scale=0.47]{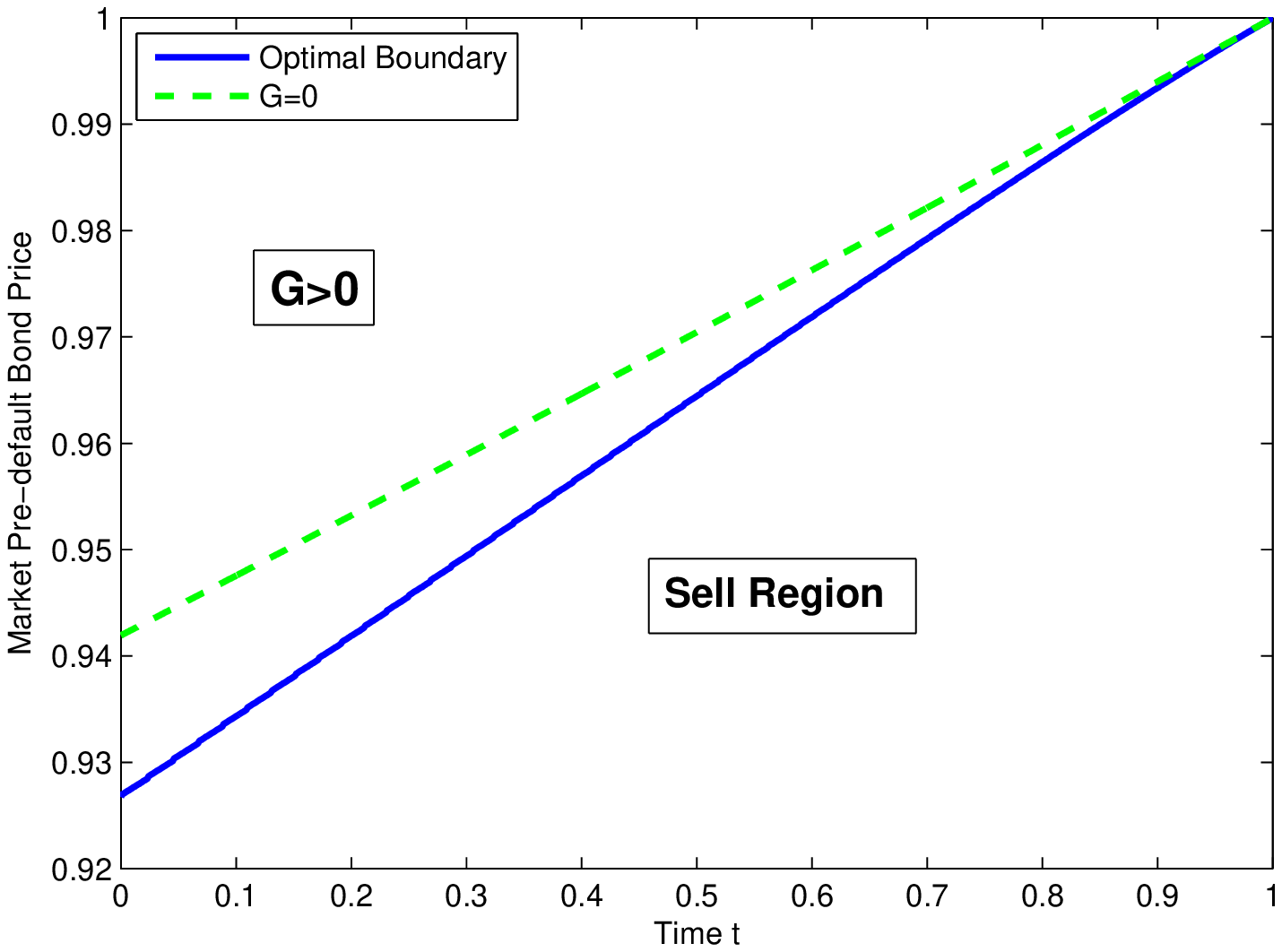}
\end{tabular}
\caption{\small{Optimal liquidation boundary in terms of market pre-default bond price under OU dynamics. We take $T=1$, $r=0.03$, $\sigma=0.02$, $\mu=\tilde{\mu}=2$, and $\theta_\lambda=\tilde{\theta}_\lambda=0.015$. \emph{Left panel}: When $\kappa_\lambda=0.2<0.3=\tilde{\kappa}_\lambda$, the optimal boundary increases from $0.958$ to $1$ over time. \emph{Right panel}: When $\kappa_\lambda=0.3>0.2=\tilde{\kappa}_\lambda$, the optimal boundary increases from $0.927$ to $1$ over time. The dashed straight line is defined by $G=0$, and we have $G\le 0$ in both sell regions.}}
\label{fig1}
\end{figure}

\begin{figure}[ht]
\centering
\begin{tabular}{cc}
\includegraphics[scale=0.47]{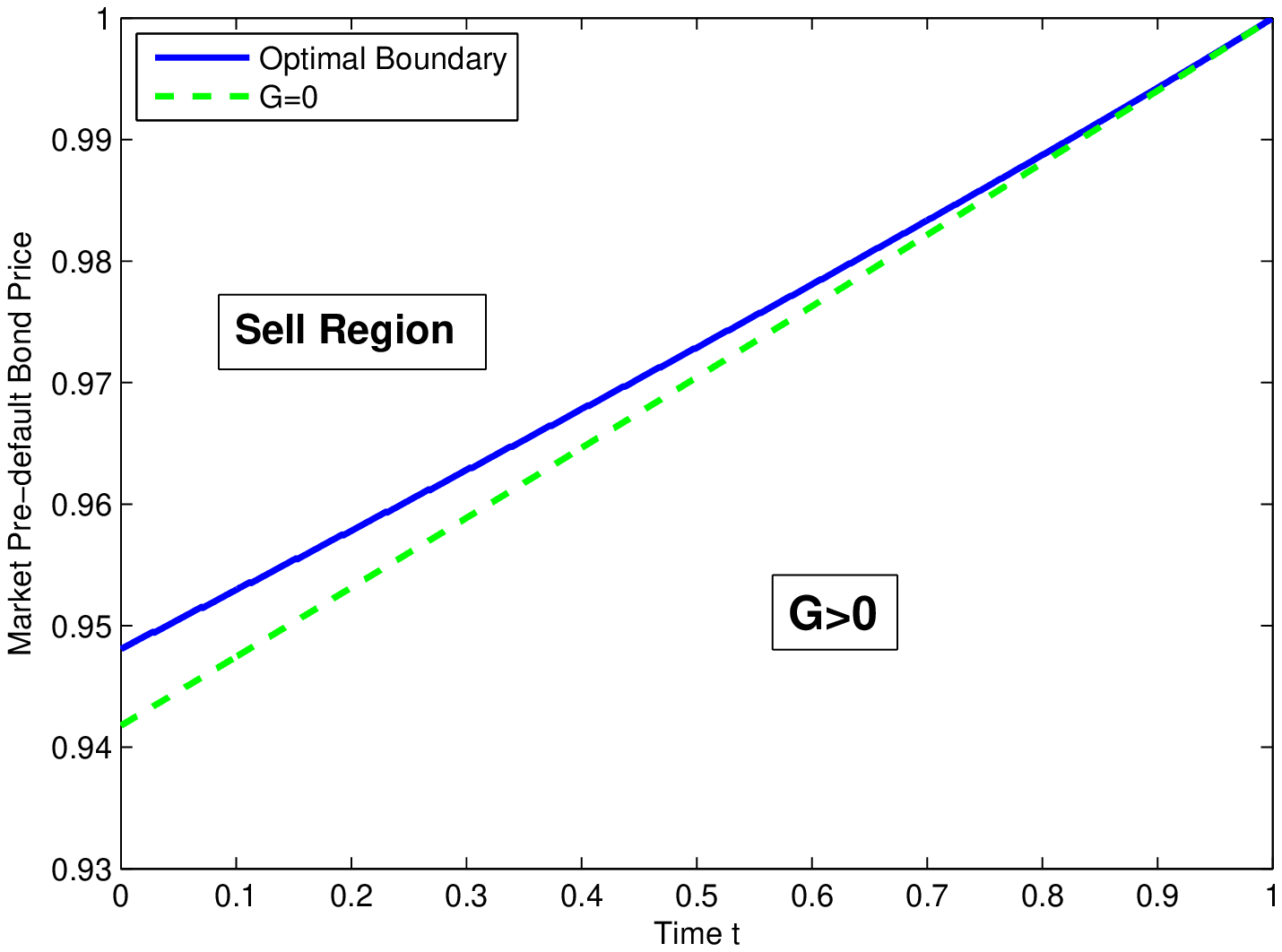} &
\includegraphics[scale=0.47]{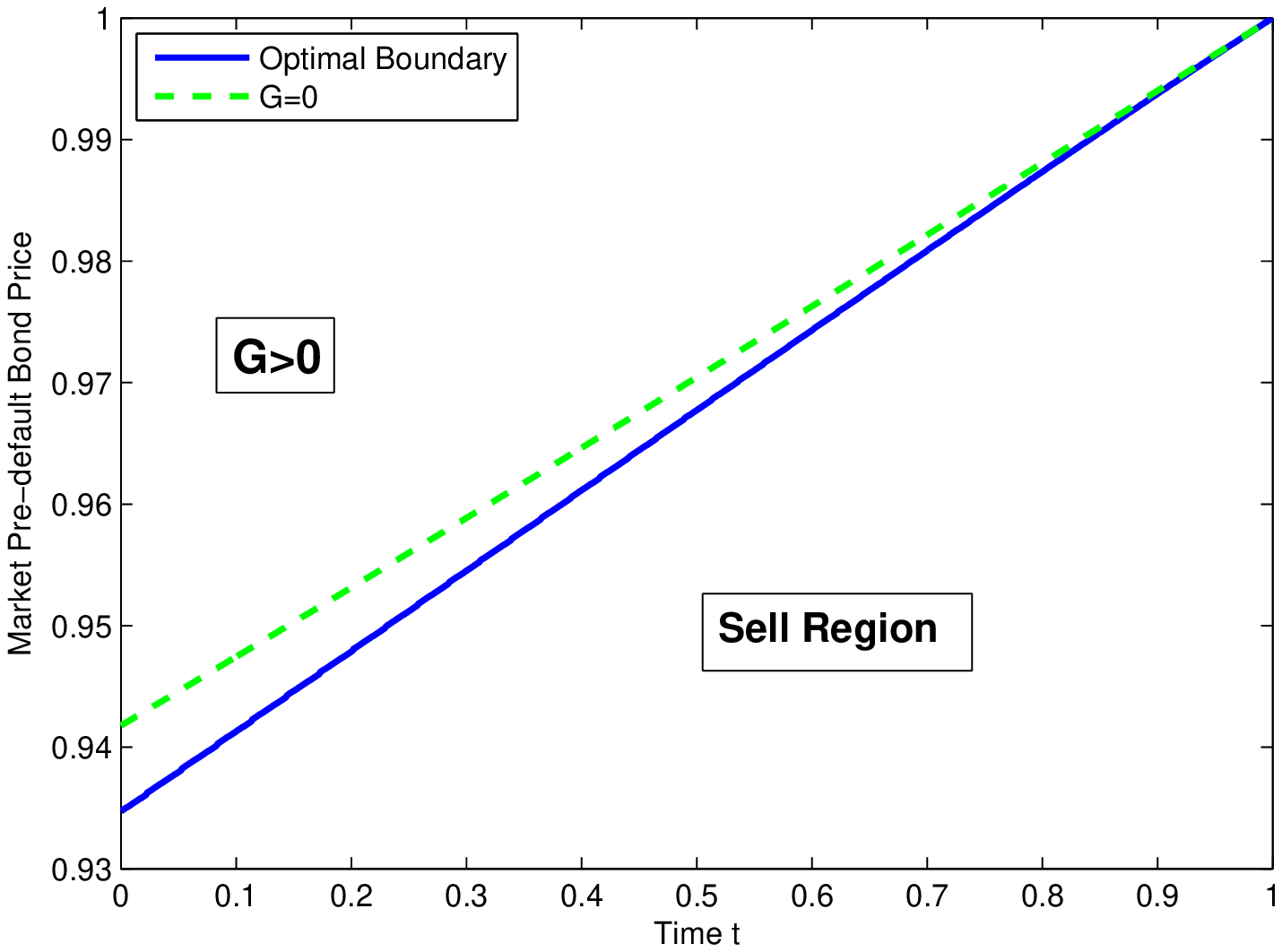} \\
\end{tabular}
\caption{\small{Optimal liquidation boundary in terms of market pre-default bond price under CIR dynamics. We take $T=1$, $r=0.03$, $\sigma=0.07$, $\mu=\tilde{\mu}=2$,  and $\theta_\lambda=\tilde{\theta}_\lambda=0.015$. \emph{Left panel}: When $\kappa_\lambda=0.2<0.3=\tilde{\kappa}_\lambda$, the optimal boundary increases from $0.948$ to $1$ over time. \emph{Right panel}: When $\kappa_\lambda=0.3>0.2=\tilde{\kappa}_\lambda$, the optimal boundary increases from $0.935$ to $1$ over time.}}
\label{fig2}
\end{figure}

As a result, the drift function is given by
    \begin{align}
G^0(t,{\bf x})=\big[\sum_{i=1}^n \big([B_i(T-t)(\tilde{\kappa}_i-\kappa_i)-(\tilde{\mu}-\mu)w^{\lambda}_i ]x_i+B_i(T-t)(\kappa_i\theta_i-\tilde{\kappa}_i\tilde{\theta}_i)\big)\big]C^0(t,{\bf x}),\nonumber
\end{align}which is again linear in terms of $C^0(t,{\bf x})$.

To illustrate the optimal liquidation strategy, we consider a numerical example where interest rate is constant, ${\bf X}$=$\lambda$, and ${\bf w}^{\lambda}=\frac{1}{\mu}$. The benchmark specifications for the market default intensity $\lambda$ in the CIR dynamics are $T=1$, $\mu=2$, $\kappa_\lambda=0.2$, $\theta_\lambda=0.015$, $r=0.03$, and $\sigma=0.07$ based on the estimates from \cite{Driessen2005,Duffee1999}. Like in the OU model, we can again express the sell region and delay region in terms of the pre-default market price $C^0$; see Figure \ref{fig2}.

\subsection{Recovery of Treasury and Market Value}\label{RT}
Extending the preceding analysis on defaultable bonds, we incorporate two principle ways of modeling recovery: the recovery of treasury or  market value.

By the recovery of treasury, we assume  that a recovery of $c$ times the value of the equivalent default-free bond is paid upon default. Therefore, the market pre-default bond price function is
\begin{equation}
C^{RT}(t,{\bf x})=(1-c)C^0(t,{\bf x})+c\beta(t,{\bf x}), \label{C_RT}
\end{equation}
where $\beta(t,{\bf x}):=\E^\Q\big\{e^{-\int_t^T \! r_v  dv}|{\bf X}_t={\bf x}\big\}$ is the  equivalent default-free bond price. Then, applying Theorem \ref{thm_main} with the quadruple $(1, 0, c\beta, \tau_d)$, we obtain the corresponding drift function:
\begin{align}
G^{RT}(t,{\bf x})=-\big(\nabla_x C^{RT}(t,{\bf x})\big)^T\Sigma(t,{\bf x}){\boldsymbol \phi}^{\tilde{\Q},\Q}(t,{\bf x})+(c-1)\big(\tilde{\mu}(t,{\bf x})-\mu(t,{\bf x})\big)\hat{\lambda}(t,{\bf x})C^0(t,{\bf x}).\label{G_RT}
\end{align}
 If $c=0$, then $C^{RT}(t,{\bf x})=C^0(t,{\bf x})$ and  $G^{RT}$ in \eqref{G_RT} reduces to the drift function of the zero-recovery bond. If $c=1$, then $C^{RT}(t,{\bf x})=\beta(t,{\bf x})$ is the market price of a default-free bond, and risk premium discrepancy may arise only from the interest rate dynamics.

Here are  two  examples where  the drift function $G^{RT}$ in (\ref{G_RT}) can be computed explicitly.
\begin{example}
Under OU model, $C^{RT}(t,r,\lambda)$ is computed according to (\ref{C_RT}) with $C^0(t,r,\lambda)$ in ($\ref{C_OU}$) and $\beta(t,r, \lambda)=e^{\bar{A}(T-t)-B(T-t)r}$,
where $\bar{A}(s)=\int_0^s \! \big[\frac{1}{2}\sigma_r^2 B^2(z)-\kappa_r\theta_r B(z)\big] \, dz$ and $B(s)$ is defined in (\ref{B_OU}).
\end{example}
\begin{example}
Under the multi-factor CIR model, $C^{RT}(t,{\bf x})$ is found again from (\ref{C_RT}), where $C^0(t,{\bf x})$ is given in ($\ref{C_CIR}$), and $\beta(t,{\bf x})$ is computed from ($\ref{C_CIR}$) with ${\bf w}^\lambda={\bf 0}$ in ($\ref{A_CIR}$) and ($\ref{Xi}$).
\end{example}

As for the recovery of market value, we assume that at default the recovery is $c$ times the pre-default value $C^{RMV}_{\tau_d -}$. The market pre-default price is given by
\begin{align}
C^{RMV}(t,{\bf X}_t)=\E^\Q\big\{e^{-\int_t^T \! (r_v+(1-c)\lambda_v)  dv}|\F_t\big\},\quad 0\leq t\leq T.
\end{align}
The  corresponding drift function can be obtained by applying the quadruple $(1, 0, cC^{RMV}, \tau_d)$ to  Theorem \ref{thm_main}.

\begin{example}
Under the OU model in Section \ref{sect-OU}, the price function
$C^{RMV}(t,r,\lambda)$ is given by
\begin{equation}
C^{RMV}(t,r,\lambda)=e^{\hat{A}(T-t)-B(T-t)r-\hat{D}(T-t)\lambda}, \nonumber
\end{equation}
where $B(s)$ is defined in (\ref{B_OU}),
\begin{align}
\hat{D}(s)&=\frac{(1-c)(1-e^{-\kappa_\lambda s})}{\kappa_\lambda}, ~\text{ and}\nonumber\\
\hat{A}(s)&=\int_0^s \! \big[\frac{1}{2}\sigma_r^2 B^2(z)+\rho\mu\sigma_r\sigma_\lambda B(z)\hat{D}(z)+\frac{1}{2}\mu^2\sigma_\lambda^2 \hat{D}^2(z)-\kappa_r\theta_r B(z)-\mu\kappa_\lambda\theta_\lambda \hat{D}(z)\big] dz.\nonumber
\end{align}
\end{example}

\begin{example}
Under the multi-factor CIR model, $C^{RMV}(t,{\bf x})$ admits the same
formula as ($\ref{C_CIR}$) but with ${\boldsymbol w^{\lambda}}$
replaced by $(1-c){\boldsymbol w^{\lambda}}$ in ($\ref{A_CIR}$) and ($\ref{Xi}$).
\end{example}

\subsection{Optimal Liquidation of CDS}
 In this section we consider optimally liquidating a digital CDS position.  The investor is a protection buyer who pays a fixed premium to the protection seller from time $0$ until default or maturity $T$, whichever comes first. The premium rate $p^m_0$, called the market spread, is specified at contract inception. In return, the protection buyer will receive  \$$1$  if  default occurs at or before $T$. The liquidation of the CDS position at time $t$ can be achieved by entering a CDS contract as a protection seller with the same credit reference and same maturity $T$ at the prevailing market spread $p^m_t$.
By definition, the prevailing market spread $p^m_t$  makes the values of two legs equal at time $t$, i.e.
\begin{align}
\E^\Q\big\{\int_t^T \! e^{-\int_t^u \! r_v  dv}p^m_t\1_{\{u<\tau_d\}} du|\G_t\big\}=\E^\Q\big\{e^{-\int_t^{\tau_d} \! r_v  dv}\1_{\{t<\tau_d\leq T\}}|\G_t\big\}.\label{marketspread}\end{align}
If the liquidation occurs at time $t$, she receives the premium at rate $p^m_t$ and pays the premium at rate $p^m_0$ until default or maturity $T$. If default occurs,  then the default payments from both CDS contracts will cancel. Considering the resulting expected cash flows and \eqref{marketspread},  the mark-to-market value of the  CDS is given by
\begin{align}
\E^\Q\big\{\int_t^T \! e^{-\int_t^u \! r_v  dv}(p^m_t-p^m_0)\1_{\{u<\tau_d\}} du|\G_t\big\}=&\1_{\{t<\tau_d\}}\E^\Q\big\{\int_t^T \! e^{-\int_t^u \! (r_v+\lambda_v)  dv}(\lambda_u-p^m_0) du|\F_t\big\}\notag\\
=&:\1_{\{t<\tau_d\}}C^{CDS}(t,{\bf X}_t).\label{CCDSprice}
\end{align}

For CDS, we apply the quadruple $(0, -p^m_0, 1, \tau_d)$ to   Theorem \ref{thm_main} and obtain the drift function:
\begin{align}
G^{CDS}(t,{\bf x})=-\big(\nabla_x C^{CDS}(t,{\bf x})\big)^T\Sigma(t,{\bf x}){\boldsymbol \phi}^{\tilde{\Q},\Q}(t,{\bf x})+\big(1-C^{CDS}(t,{\bf x})\big)\big(\tilde{\mu}(t,{\bf x})-\mu(t,{\bf x})\big)\hat{\lambda}(t,{\bf x}).\label{G_CDS}
\end{align}
If there is no discrepancy over mark-to-market risk premium, i.e. ${\boldsymbol
\phi}^{\tilde{\Q},\Q}(t,{\bf x})={\bf 0}$, then the sign of $G^{CDS}$ is determined by $\tilde{\mu}(t,{\bf x})-\mu(t,{\bf x})$ since  $C^{CDS}\leq1$. From this we infer that higher event risk premium (relative to market) implies delayed liquidation.

In general, the market pre-default value $C^{CDS}$ can be solved by PDE (\ref{PDE_C_general}).
If the state vector ${\bf X}$ admits OU or CIR dynamics,  $C^{CDS}$, and thus  $G^{CDS}$,  is given in closed form, as illustrated in the following two examples.

\begin{example} \label{ex_CDS_OU}
Under the OU dynamics, the pre-default value of CDS (see \eqref{CCDSprice}) is given by the following integral:
\begin{align}
C^{CDS}(t,r,\lambda)=\int_t^T \!C^0(t,r,\lambda; u)\bigg[\lambda e^{-\kappa_\lambda(u-t)}+\int_t^u \! e^{-\kappa_r(u-s)}g(s,u) ds -p^m_0 \bigg]  du\nonumber,
\end{align}
where $C^0(t, r,\lambda; u)$ is given by ($\ref{C_OU}$) with $T=u$ and
\begin{align}
g(s,u):=\mu\kappa_\lambda\theta_\lambda-\rho\mu\sigma_r\sigma_\lambda\frac{1-e^{-\kappa_r(u-s)}}{\kappa_r}-(\mu\sigma_\lambda)^2
\frac{1-e^{-\kappa_\lambda(u-s)}}{\kappa_\lambda}.\nonumber
\end{align}
\end{example}
\begin{example} \label{ex_CDS_CIR}
Under the multi-factor CIR dynamics, the pre-default value of CDS is given by the following integral:
\begin{align}
C^{CDS}(t,{\bf x})=\int_t^T \!C^0(t,{\bf x}; u)\bigg[ \displaystyle\sum_{i=1}^n \bigg(\mu w^\lambda_i\big(\kappa_i\theta_iB_i(u-t)+ B_i^\prime(u-t)x_i\big)\bigg) -p^m_0 \bigg]  du\nonumber,
\end{align}
where $C^0(t,{\bf x}; u)$ is given in  ($\ref{C_CIR}$) with $T=u$ and $B_i(s)$ in (\ref{Xi}). See Chapter $7$ of \cite{Schonbucher2003}.
\end{example}
\begin{example} For a \emph{forward}  CDS with  start date $T_a<T$, the protection buyer pays premium at rate $p_a$ from $T_a$ until $\tau_d$ or maturity $T$, and receives $1$ if $\tau_d \in [T_a, T]$.  By direct computation,  the pre-default market value  is $C^{CDS}(t,{\bf x};T)-C^{CDS}(t,{\bf x};T_a)$,   $t<T_a$. Consequently, closed-form formulas  for the drift function  are available under OU or CIR dynamics  by Examples \ref{ex_CDS_OU} and \ref{ex_CDS_CIR}.
\end{example}

We consider a numerical example where interest rate is constant and state vector ${\bf X}$=$\lambda$ follows the CIR dynamics. We assume that the investor agrees with the market on all parameters except the speed of mean reversion for default intensity. In the left panel of Figure \ref{fig3} with $\kappa_\lambda=0.2<0.3=\tilde{\kappa}_\lambda$, the optimal liquidation strategy is to sell as soon as the market CDS value reaches an upper boundary. In the case with   $\kappa_\lambda=0.3>0.2=\tilde{\kappa}_\lambda$ (see Figure \ref{fig3} (right)), the sell region  is below the continuation region.

\begin{figure}[ht]
\centering
\begin{tabular}{cc}
\includegraphics[scale=0.458]{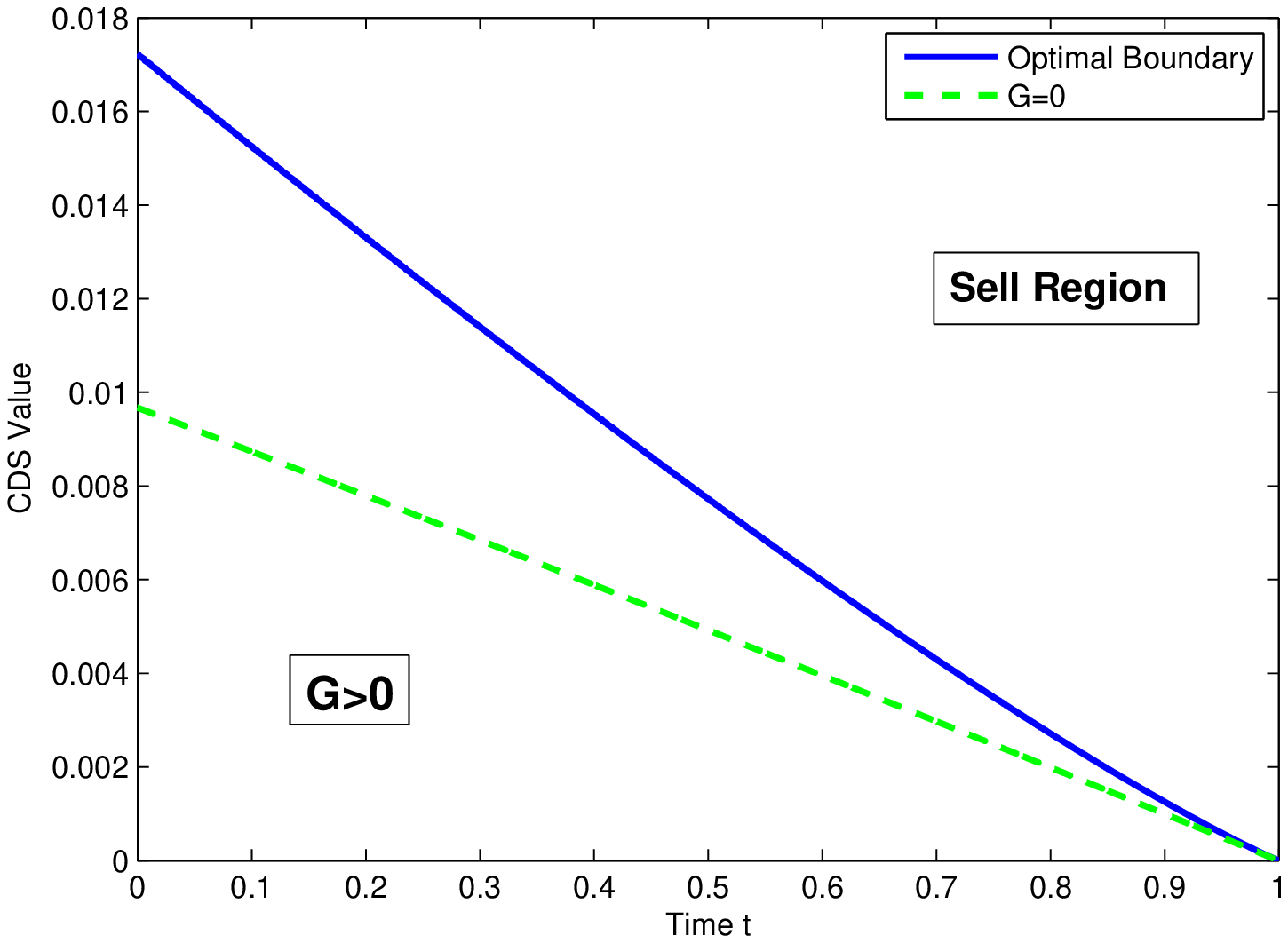} &
\includegraphics[scale=0.458]{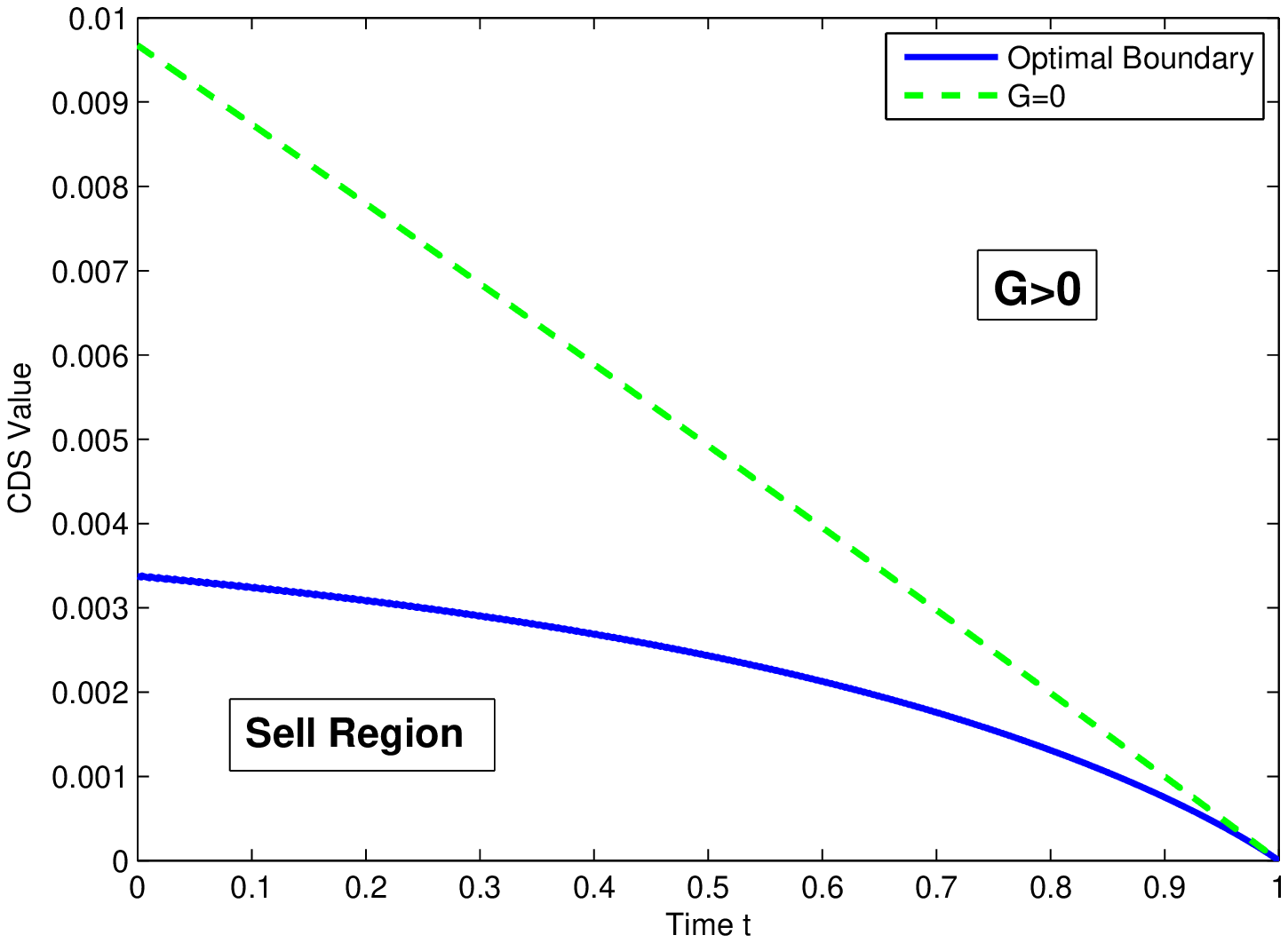} \\
\end{tabular}
\caption{\small{Optimal liquidation boundary in terms of market pre-default CDS value under CIR dynamics. We take $T=1$, $r=0.03$, $\sigma=0.07$, $p^m_0=0.02$, $\mu=\tilde{\mu}=2$, and $\theta_\lambda=\tilde{\theta}_\lambda=0.015$. \emph{Left panel}: When $\kappa_\lambda=0.2<0.3=\tilde{\kappa}_\lambda$, liquidation occurs at an upper boundary that decreases from $0.0172$ to $0$ over $t\in[0,1]$. \emph{Right panel}: When $\kappa_\lambda=0.3>0.2=\tilde{\kappa}_\lambda$, the CDS is liquidated at a lower liquidation boundary, which  decreases from $0.00338$ to $0$ over time. In both cases, the dashed line defined by $G=0$ lies within the continuation region. }}
\label{fig3}
\end{figure}

\begin{remark} As a  straightforward generalization under our framework,  one can replace the unit payment at default by $R^{CDS}(\tau_d, {\bf X_{\tau_d}})$. Then,  we can apply the quadruple $(0, -p^m_0,R^{CDS} , \tau_d)$ to  Theorem \ref{thm_main}, and obtain the same drift function  $G^{CDS}$ in  \eqref{G_CDS} except with $1$ replaced by $R^{CDS}(t, {\bf x})$.
\end{remark}

\subsection{Jump-Diffusion Default Intensity} We can extend our analysis to incorporate jumps to the stochastic state vector. To illustrate this, suppose   the default intensity and interest rate are driven by a $n$-dimensional state vector ${\bf X}'$  with  the affine jump-diffusion dynamics:
\begin{align}
d{\bf X}'_t=a(t,{\bf X}'_t)dt+\Sigma(t,{\bf X}'_t)d{\bf W}^\P_t+ d{\bf J}_t,
\end{align}
where ${\bf J}=(J^1,\ldots,J^n)^T$ is a  vector of $n$ independent pure jump
processes taking values in $\R^n$. Under historical measure $\P$, we assume Markovian jump intensity of the form $\hat{\bf {\Lambda}}(t,{\bf X}'_t)=(\hat{\Lambda}^1(t,{\bf X}'_t),\ldots,\hat{\Lambda}^n(t,{\bf X}'_t))^T$ for ${\bf J}$. All random jump sizes $(Y^i_j)_{ij}$ of ${\bf J}$ are independent, and for each  $J^i$,  the associated jump sizes $Y^i_1, Y^i_2,\ldots$ have  a common probability density
function $\hat{f}^i$.

The default intensity of defaultable security is given by $\hat{\lambda}(t,{\bf X}'_t)$ for some positive measurable function $\hat{\lambda}(\cdot, \cdot)$, and the default counting process associated with default time $\tau_d$ is denoted by $N_t=\1_{\{t\geq\tau_d\}}$. We denote $(\mathcal{G}_t)_{0\le t\le T}$ to be the full filtration generated by ${\bf W}^\P, {\bf J}$, and $\tau_d$.

We define a market pricing measure $\Q$ in terms of the mark-to-market  risk premium ${\boldsymbol \phi}^{\Q, \P}$ and event risk premium $\mu$, which are  Markovian and  satisfy $\int_0^T  |\!| {\boldsymbol \phi}^{ {\Q}, \P}_u |\!|^2 du<\infty$ and  $\int_0^T \mu_u\hat{\lambda}_u du<\infty$. Due to the presence of ${\bf J}$, the market measure $\Q$ can  scale the jump intensity of ${\bf J}$ by the positive Markovian factors  $\delta_t^i = \delta^i(t,{\bf X}'_t)$, with
$\int_0^T \delta^i_u \hat{\Lambda}^i_u du<\infty$ for $i=1,\ldots,n$. Also, $\Q$ can transform the jump size distribution of ${\bf J}$ by a function $(h^i)_{i=1,\ldots,n} > 0$ satisfying $\int_0^\infty h^i(y)\hat{f}^i(y)dy=1$ for $i=1,\ldots,n$.

  The Radon-Nikodym derivative is given by
\begin{align}
\frac{d{\Q}}{d\P}\big|\mathcal{G}_t=\mathcal{E}\big(- {\boldsymbol \phi}^{\Q, \P}\! \cdot\! {\bf W}^\P\big)_t\,\mathcal{E}\big((\mu-1) M^\P\big)_t K_t^{\Q, \P}\nonumber,
\end{align}
where  first two Dol\'eans-Dade exponentials are defined in \eqref{measure_W_PQ} and \eqref{measure_M_P_Q} respectively, $M^\P_t:=N_t-\int_0^t \! (1-N_u)\hat{\lambda}_u du$ is the compensated $\P$-martingale associated with $N$, and the last term
\begin{align}
K_t^{\Q, \P}:=\prod_{i=1}^n \bigg[\exp\bigg(\int_0^t\int_{\R^n}{\big(1-\delta^i(u,{\bf X}'_u)  h^i(y)\big)\hat{\Lambda}^i(u,{\bf X}'_u)\hat{f}^i(y)d{y} du}\bigg)\prod_{j=1}^{N^{(i)}_t}\,\big(\delta^i(T_j^i,{\bf X}'_{T_j^i})h^i(Y^i_j)\big)\bigg]\notag,
\end{align}
where  $T^i_j$ is the $j$th jump time of $J^i$ and $N_t^{(i)}:=\sum_{j\geq1}\1_{\{T^i_j\leq t\}}$ is the counting process associated with $J^i$.

By  Girsanov Theorem, ${\bf W}^{\Q}_t:={\bf W}^\P_t+\int_0^t{\boldsymbol \phi}_u^{\Q,\P}du$ is a $\Q$-Brownian motion, the jump intensity of $J^i$ under $\Q$ is $\Lambda^i(t,{\bf X}'_t):=\delta^i(t,{\bf X}'_t)\hat{\Lambda}^i(t,{\bf X}'_t)$, and the jump size pdf of $J^i$ under $\Q$ is  $f^i(y):=h^i(y)\hat{f}^i(y)$. The $\Q$-dynamics of state vector ${\bf X}'$ is given by
\begin{align}
d{\bf X}'_t=b(t,{\bf X}'_t)dt+\Sigma(t,{\bf X}'_t)d{\bf W}^\Q_t+ d{\bf J}_t.
\end{align} Also incorporating the event risk premium, the  $\Q$-default intensity is  $\lambda(t,{\bf X}'_t):=\mu(t,{\bf X}'_t)\hat{\lambda}(t,{\bf X}'_t)$.

Under the investor's measure $\tilde{\Q}$, we replace $b$ with $\tilde{b}$, ${\boldsymbol \phi}^{\Q,\P}$ with ${\boldsymbol \phi}^{\tilde{\Q},\P}$, and  ${\bf W}^\Q$ with ${\bf
W}^{\tilde{\Q}}$ for the dynamics of ${\bf X}'$. For each $J^i$, the $\tilde{\Q}$-intensity  is denoted by
$\tilde{\Lambda}^i(t,{\bf X}'_t):={\tilde{\delta}}^i(t,{\bf X}'_t)\hat{\Lambda}^i(t,{\bf X}'_t)$, and the jump size pdf under $\tilde{\Q}$ is $\tilde{f}^i(y):=\tilde{h}^i(y)\hat{f}^i(y)$. With investor's event risk premium $\tilde{\mu}$, the default intensity under $\tilde{\Q}$ is $\tilde{\lambda}(t,{\bf X}'_t):=\tilde{\mu}(t,{\bf X}'_t)\hat{\lambda}(t,{\bf X}'_t)$.

The two pricing measures  $\Q$ and $\tilde{\Q}$  are related by the Radon-Nikodym derivative:
\begin{align}
\frac{d\tilde{\Q}}{d\Q}\big|\mathcal{G}_t=\mathcal{E}\big(- {\boldsymbol \phi}^{\tilde{\Q}, \Q}\! \cdot\! {\bf W}^\Q\big)_t\,\mathcal{E}\big((\frac{\tilde{\mu}}{\mu}-1) M^\Q\big)_t K_t^{\tilde{\Q}, \Q}\nonumber,
\end{align}
where $M^\Q_t:=N_t-\int_0^t \! (1-N_u)\lambda_u du$ is the compensated $\Q$-martingale associated with $N$,  the first two Dol\'eans-Dade exponentials are defined in \eqref{measure_W} and \eqref{measure_M}, and
\begin{align}
K_t^{\tilde{\Q}, \Q}:=\prod_{i=1}^n\bigg[\exp\bigg(\int_0^t\int_{\R^n}{\big(\Lambda^i(u,{\bf X}'_u)f^i(y)-\tilde{\Lambda}^i(u,{\bf X}'_u)\tilde{f}^i(y)\big)dy du}\bigg)\prod_{j=1}^{N^{(i)}_t}\,\frac{\tilde{\Lambda}^i(T^i_j,{\bf X}'_{T^i_j}) \tilde{f}^i({\bf Y}^i_j)}{\Lambda^i(T^i_j,{\bf X}'_{T^i_j})f^i({\bf Y}^i_j)}\bigg]\notag.
\end{align}
Consequently, on top of  the mark-to-market risk and event risk premia, the investor can potentially  disagree with  the market over  jump intensity and jump size distribution of ${\bf X}'$, allowing for a richer structure of price discrepancy as well as the  optimal liquidation strategy.

As in Theorem \ref{thm_main}, we compute the drift function in terms of pre-default price $C$ and default risk premia, namely,
\begin{align}
G^J(t,{\bf x})&=-\big(\nabla_{x} C(t,{\bf x})\big)^T\Sigma(t,{\bf x}){\boldsymbol \phi}^{\tilde{\Q},\Q}(t,{\bf x})+(R(t,{\bf x})-C(t,{\bf x}))\big(\tilde{\mu}(t,{\bf x})-\mu(t,{\bf x})\big)\hat{\lambda}(t,{\bf x})\nonumber\\
&+\sum_{i=1}^n \bigg(\int_{\R^n}{\big(C(t,{\bf x}+y{\bf e}_i)-C(t,{\bf x})\big)\big(\tilde{\Lambda}^i(t,{\bf x}) \tilde{f}^i(y)-\Lambda^i(t,{\bf x}) f^i(y)\big)dy}\bigg),\label{G_jump}
\end{align}
where ${\bf e}_i:=(0, \ldots, 1, \ldots,0)^T$. We observe that the first two components of $G^J$ share the same functional form as $G$ in \eqref{G_general}, though the  price function $C$ is derived from the jump-diffusion model. Even if the investor and the market assign the same mark-to-market risk and event risk premia, discrepancy over jump intensity and distribution will yield different liquidation strategies.  Under  quite general affine jump-diffusion models,  Duffie et al. \cite{Duffie2000} provide an analytical treatment of transform analysis, which can be used for the computation of our drift function.

\section{Optimal Liquidation of Credit Default Index Swaps}\label{multi}
We proceed to discuss the optimal liquidation of multi-name credit derivatives. In the literature, there exist many proposed models for modeling multiple  defaults and pricing multi-name credit derivatives. Within the intensity-based framework, one popular approach is to model each default time by the first jump of a doubly-stochastic process. The dependence among defaults can be incorporated via some common stochastic factors. This well-known bottom-up valuation framework has been studied  in \cite{duffiegarleanu, mortensen}, among many others.

As a popular alternative, the top-down approach  describes  directly the dynamics of the cumulative  credit  portfolio loss,  without detailed references to  the constituent single names. Some examples of  top-down models   include  \cite{Arnsdorff2008,Brigo2006,Ding2009,LongstaffRajan2008,Lopatin2008}. In particular,  Errais et al. \cite{Errais2010} proposed affine point processes for portfolio loss with self-exciting property to capture default clustering. For our analysis, rather than  proposing a new multi-name credit risk model, we adopt the self-exciting top-down model from \cite{Errais2010}. Also, we will  focus on the  optimal liquidation of a credit default index swap.

First, we model successive default arrivals by a counting process $(N_t)_{0\le t\le T}$, and  the accumulated portfolio loss   by $\Upsilon_t=l_1+\ldots+l_{N_t}$, with each  $l_n$ representing the random loss at the $n$th default.  Under the  historical measure $\P$, the default  intensity evolves according  to  the jump-diffusion:
\begin{align}\label{SDElamda_multi}
d\hat{\lambda}_t=\hat{\kappa}(\hat{\theta}-\hat{\lambda}_t)dt+\sigma\sqrt{\hat{\lambda}_t}\,dW_t^\P+\eta \,d\Upsilon_t,
\end{align}
where $W^\P$ is a standard $\P$-Brownian motion. We assume that the random losses $(l_n)$ are independent with an identical probability density function $\hat{m}$ on $(0,\infty)$. According to the last term in \eqref{SDElamda_multi}, each  default arrival will increase  default intensity $\hat{\lambda}$ by the loss at default scaled by the positive parameter $\eta$. This term captures default clustering observed in the multi-name credit derivatives as pointed out in \cite{Errais2010}.  We assume a constant risk-free interest rate $r$ for simplicity, and denote $(\mathcal{H}_t)_{0\le t\le T}$ to be the full filtration   generated by $N$, $\Upsilon$, and $W^\P$.

The market measure $\Q$ is characterized by several key components. First, the market's mark-to-market risk premium is assumed to  be  of the form \begin{align}\label{topdownphi}{\phi}_t^{\Q,\P}=\frac{\hat{\kappa}(\hat{\theta}-\hat{\lambda}_t)-\kappa(\theta-\hat{\lambda}_t)}{\sigma \sqrt{\hat{\lambda}_t}}\end{align}such that the default intensity in \eqref{SDElamda_multi} preserves mean-reverting dynamics with different parameters $\kappa$ and $\theta$ under the market measure $\Q$ (see \cite{Azizpour2008,Jarrow2005} for similar specifications). Secondly, we assume that the $\Q$-default intensity is  $\lambda_t:=\mu\hat{\lambda}_t$, with a positive constant event risk premium. Thirdly, the distribution of random losses can be  scaled  under $\Q$. Specifically, we assume that under $\Q$ the losses $(l_n)$ admit  the pdf $m(z):=h(z)\hat{m}(z)$, for some strictly positive function $h$ with $\int_0^\infty h(z)\hat{m}(z)dz=1$ . Then,     the Radon-Nikodym derivative associated with $\Q$ and $\P$ is
\begin{align}
\frac{d\Q}{d\P}\big|\mathcal{H}_t=\mathcal{E}\big(- { \phi}^{\Q, \P} W^\P\big)_t \hat{K}^{\Q, \P}_t\,,
\end{align}
where $\mathcal{E}\big(- {\phi}^{\Q, \P} W^\P\big)$ is defined in \eqref{measure_W_PQ}, and
\begin{align}
\hat{K}_t^{\Q, \P}:=\exp\bigg(\int_0^t\int_0^\infty{\big(1-\mu h(z)\big)\hat{\lambda}_u\hat{m}(z)dz du}\bigg)\prod_{i=1}^{N_t}\,\big(\mu h(l_i)\big).
\end{align}
Under the market pricing measure $\Q$, the $\Q$-default intensity evolves according  to:
\begin{align}\label{SDElamda_multi_Q}
d {\lambda}_t=\kappa(\mu \theta- {\lambda}_t)dt+\sigma\sqrt{ \mu{\lambda}_t}\,dW_t^\Q+\mu\eta\, d\Upsilon_t,
\end{align}
where ${W}^{\Q}_t:={W}^\P_t+\int_0^t{\phi}_u^{\Q,\P}du$ is a standard $\Q$-Brownian motion.  Similarly, we can define the investor's pricing measure $\tilde{\Q}$ through the investor's mark-to-market risk premium ${\phi}^{\tilde{\Q},\P}$ as in \eqref{topdownphi} with   parameters $\tilde{\kappa}$ and $\tilde{\theta}$; default intensity $\tilde{\lambda}_t = \tilde{\mu}\hat{\lambda}_t$ with constant event risk premium $\tilde{\mu}$; and loss scaling  function $\tilde{h}$ so that the loss pdf $\tilde{m}(z) = h(z) \hat{m}(z)$.

The credit default index swap is written on a standardized portfolio of $H$ reference entities, such as single-name default swaps, with same notational normalized to $1$ and same maturity $T$. The investor is a protection buyer who pays at the premium rate $p^m_0$ in return for default payments over $(0,T]$. Here, the default payment is assumed to be paid at the time when default occurs, and the premium payment is paid continuously with premium notational equal to $H-N_t$.

The market's cumulative value of the credit default index swap for the protection buyer is equal to the difference between the market values of the default payment leg and premium leg, namely,
\begin{align}
P^{CDX}_t=\E^\Q\big\{ \int_{(0,T]} e^{-r(u-t)} \ d \Upsilon_u\,|\mathcal{H}_t \big\}- \E^\Q\big\{p^m_0\int_{(0,T]} e^{-r(u-t)} (H-N_u)\,du \,|\mathcal{H}_t \big\}, \quad t\leq T. \label{PCDX}
\end{align}
Hence, similar to  \eqref{V_eq}, the protection buyer solves the following optimal stopping problem:
\begin{align}
V^{CDX}_t=\esssup_{\tau \in \mathcal{T}_{t,T}}~\E^{\tilde{\Q}}\big\{e^{-r(\tau-t)}P^{CDX}_\tau\,|\mathcal{H}_t\big\}\label{V_multi}.
\end{align}
The associated  delayed liquidation premium  is defined by
\begin{align}\label{liq_multi}
L^{CDX}_t=V^{CDX}_t-P^{CDX}_t.\end{align}

The derivation of the optimal liquidation strategy  involves computing the market's ex-dividend value, defined by
\begin{align}
{C}^{CDX}_t=&\E^\Q\big\{ \int_{(t,T]} e^{-r(u-t)} \ d \Upsilon_u\,|\mathcal{H}_t \big\}- \E^\Q\big\{p^m_0\int_{(t,T]} e^{-r(u-t)} (H-N_u)\,du \,|\mathcal{H}_t \big\}.\label{ex_ind}
\end{align}

\begin{proposition} \label{ind_prop}
The market's ex-dividend value of the credit default index swap  in \eqref{ex_ind} can be expressed as ${C}^{CDX}_t = {C}^{CDX}(t,\lambda_t, N_t)$, where
\begin{align}\label{Cind22}
{C}^{CDX}(t,\lambda, n)=k_2(t,T) \lambda+k_1(t,T) n+k_0(t,T),
\end{align}
for $t\le T$, with  coefficients
\begin{align}
&k_2(t,T)=(cr+p^m_0)\bigg(\frac{e^{-(\rho+r)(T-t)}}{\rho(\rho+r)}-\frac{e^{-r(T-t)}}{\rho r} +\frac{1}{r(\rho+r)}\bigg)+\frac{c e^{-r(T-t)}}{\rho }\big(1-e^{-\rho(T-t)}\big), \\
&k_1(t,T)=\frac{p^m_0\big(1-e^{-r(T-t)}\big)}{r}, \\
& k_0(t,T)=\frac{\kappa\mu\theta}{\rho}\bigg((rc+p^m_0)\big[e^{-r(T-t)}\big( -\frac{e^{-\rho(T-t)}}{\rho(\rho+r)}-\frac{T-t}{r}-\frac{1}{r^2}+\frac{1}{r\rho}\big)+\frac{\rho}{r^2(r+\rho)}\big]\notag \\
&~~~~~~~~~\,+ce^{-r(T-t)}\big(\frac{e^{-\rho(T-t)}-1}{\rho}+T-t\big)\bigg)-\frac{p^m_0 H}{r}\big(1-e^{-r(T-t)}\big),
\end{align}
and constants \begin{align}\label{constants1}c=\int_0^{\infty}z m(z)dz, \qquad  \text{and }\qquad  \rho=\kappa-\mu\eta c.\end{align}
\end{proposition}

\begin{proof}
Using integration by parts, we re-write  the market's ex-dividend value as
\begin{align}
{C}^{CDX}_t=e^{-r(T-t)} \ \E^\Q\{\Upsilon_T|\mathcal{H}_t \}-\Upsilon_t+\int_t^T {e^{-r(u-t)} \ \big[r\E^\Q\{\Upsilon_u |\mathcal{H}_t\}-p^m_0\big(H-\E^\Q\{N_u |\mathcal{H}_t\}\big)\big]\,du}.\label{integration}
\end{align}
Hence, the computation of ${C}^{CDX}$ involves calculating $\E^\Q\{N_u\,|\,\mathcal{H}_t \}$ and $\E^\Q\{\Upsilon_u\,|\,\mathcal{H}_t \}$, $u\geq t$. Since default intensity $\lambda$ follows a square-root jump-diffusion dynamics, these conditional expectation admit the closed-form formulas (see e.g. Section $4.3$ of \cite{Errais2010}):
\begin{align}
\E^\Q\{N_u\,|\,\lambda_t&=\lambda, N_t=n, \Upsilon_t=\upsilon\}=\mathcal{A}(t,u)+\mathcal{B}(t,u)\lambda+n,  \label{cond_N}\\
\E^\Q\{\Upsilon_u\,|\,\lambda_t&=\lambda, N_t=n, \Upsilon_t=\upsilon\}=c\mathcal{A}(t,u)+c\mathcal{B}(t,u)\lambda+\upsilon,  \label{cond_L}
\end{align} for $ t \le u\le T$, where \begin{align}
\mathcal{A}(t,u)&=\frac{\kappa\mu\theta}{\kappa-\mu\eta c} \big(\frac{e^{-(\kappa-\mu\eta c) (u-t)}-1}{\kappa-\mu\eta c}+u-t\big),\label{A_moment}\\
\mathcal{B}(t,u)&=\frac{1}{\kappa-\mu\eta c}(1-e^{-(\kappa-\mu\eta c) (u-t)}).\label{B_moment}
\end{align}
Here, $c$ is the market's expected loss at default given in \eqref{constants1}.  Substituting \eqref{cond_N} and \eqref{cond_L} into \eqref{integration}, we obtain the closed-form formula for market's ex-dividend value in \eqref{Cind22}.
\end{proof}

As a result, the ex-dividend value ${C}^{CDX}$ is linear in the default intensity $\lambda_t$ and number of defaults $N_t$.
Next, we characterize the optimal corresponding liquidation premium and   strategy.

\begin{theorem} \label{thm_multi} Under the top-down credit risk model in \eqref{SDElamda_multi}, the delayed liquidation premium associated with the credit default index swap is given by\begin{align}
L^{CDX}(t,\lambda)=\sup_{\tau \in \mathcal{T}_{t,T}}~\E^{\tilde{\Q}}\big\{\int_t^{\tau}  e^{-r(u-t)}G^{CDX}(u,\lambda_u) du \,|\,\lambda_t=\lambda \big\},\label{L_multi}
\end{align}
where
\begin{align}
G^{CDX}(t,\lambda)&=\bigg(\big(\mu\eta k_2(t,T)+1\big)(\frac{\tilde{\mu}\tilde{c}}{\mu}-c) +k_1(t,T)(\frac{\tilde{\mu}}{\mu}-1)-k_2(t,T)(\tilde{\kappa}-\kappa)\bigg)\lambda\notag\\
&~~+k_2(t,T)\mu(\tilde{\kappa}\tilde{\theta}-\kappa\theta),\label{G_multi}
\end{align}
with $\tilde{c}:=\int_0^{\infty}z \tilde{m}(z)dz$.
If $G^{CDX}(t,\lambda)\geq0$ $\forall (t,\lambda)$, then it is optimal to delay the liquidation till  maturity $T$. If $G^{CDX}(t,\lambda)\leq0$ $\forall (t,\lambda)$, then  it is optimal to sell immediately.
\end{theorem}

\begin{proof}  In view of the definition of $L^{CDX}$ in \eqref{liq_multi}, we consider the dynamics of $P^{CDX}$. First, it follows from \eqref{PCDX} and \eqref{ex_ind} that
\begin{align}
e^{-r(u-t)}P^{CDX}_u=e^{-r(u-t)}{C}^{CDX}_u+\int_{(0,u]} e^{-r(v-t)} \ \big( d \Upsilon_v-p^m_0(H-N_v)dv\big).\label{diff_ex_cum}
\end{align}
Using \eqref{diff_ex_cum} and the fact that $e^{-rt}P^{CDX}_t$ is $\Q$-martingale (whose SDE must have no drift), we apply Ito's lemma to get
\begin{align}
&e^{-r(\tau-t)}P^{CDX}_\tau-P^{CDX}_t\notag\\
=&\int_t^\tau e^{-r(u-t)}{\frac{\partial {C}^{CDX}}{\partial \lambda}}(u,\lambda_u, N_u) \sigma\sqrt{\mu\lambda_u} dW^\Q_u\nonumber\\
&+\bigg[\sum_{t<u\leq\tau} e^{-r(u-t)}(\Upsilon_u-\Upsilon_{u_-})-\int_t^\tau\int_0^\infty{e^{-r(u-t)}z m(z)\lambda_u\,dz du}\bigg]\notag\\
&+\bigg[\sum_{t<u\leq\tau} e^{-r(u-t)}\big({C}^{CDX}(u,\lambda_u,N_u)-{C}^{CDX}(u,\lambda_{u-},N_{u-})\big)\notag\\
&-\int_t^\tau\int_0^\infty{e^{-r(u-t)}\big({C}^{CDX}(u,\lambda_{u}+\mu\eta z,N_{u}+1)-{C}^{CDX}(u,\lambda_{u},N_{u})\big) m(z)\lambda_u\,dz du}\bigg]\label{eqn11}\\
=&\int_t^\tau e^{-r(u-t)}\big({\frac{\partial {C}^{CDX}}{\partial \lambda}}(u,\lambda_u, N_u) \sigma\sqrt{\mu\lambda_u} dW^{\tilde{\Q}}_u+G^{CDX}(u,\lambda_u, N_u) du\big)\nonumber\\
&+\bigg[\sum_{t<u\leq\tau} e^{-r(u-t)}(\Upsilon_u-\Upsilon_{u_-})-\int_t^\tau\int_0^\infty{e^{-r(u-t)}z \tilde{m}(z)\tilde{\lambda}_u\,dz du}\bigg]\notag\\
&+\bigg[\sum_{t<u\leq\tau} e^{-r(u-t)}\big({C}^{CDX}(u,\lambda_u,N_u)-{C}^{CDX}(u,\lambda_{u-},N_{u-})\big)\notag\\
&-\int_t^\tau\int_0^\infty{e^{-r(u-t)}\big({C}^{CDX}(u,\lambda_{u}+\mu\eta z,N_{u}+1)-{C}^{CDX}(u,\lambda_{u},N_{u})\big) \tilde{m}(z)\tilde{\lambda}_u\,dz du}\bigg],\label{int_measure_change}
\end{align} for $t \leq \tau\leq T$,  where
\begin{align}
G^{CDX}(t,\lambda,n)&:=\frac{\partial {C}^{CDX}}{\partial \lambda}(t,\lambda, n)\big((\tilde{\kappa}\tilde{\theta}-\kappa\theta)\mu-(\tilde{\kappa}-\kappa)\lambda\big)\label{G_n}
\\
&+\int_0^\infty{\big(z+{C}^{CDX}(t,\lambda+\mu\eta z,n+1)-{C}^{CDX}(t,\lambda,n)\big)(\frac{\tilde{\mu}}{\mu}\tilde{m}(z)-m(z)\big)\lambda dz \, }.\notag\end{align}
Note  that the two  compensated $\Q$-martingale terms in  \eqref{eqn11}  account for, respectively, losses and changes in $C^{CDX}$ value due to default arrivals. The second equation \eqref{int_measure_change} follows from change of measure from $\Q$ to $\tilde{\Q}$.

By Proposition \ref{ind_prop},   the terms $\frac{\partial {C}^{CDX}}{\partial \lambda}$ and ${C}^{CDX}(t,\lambda+\mu\eta z,n+1)-{C}^{CDX}(t,\lambda,n)$ do \emph{not} depend on $n$. Consequently, $G^{CDX}$ does not depend on $n$, and admits the closed-form formula \eqref{G_multi}  upon a substitution of  \eqref{Cind22} into \eqref{G_n}.

By taking the expectation on both sides of \eqref{int_measure_change} under $\tilde{\Q}$, the delayed liquidation premium $L^{CDX}$ satisfies \eqref{L_multi} and depends only on $t$ and $\lambda$. If $G^{CDX}\ge 0$, then the integrand in \eqref{L_multi} is positive a.s. and therefore the largest possible stopping time $T$ is optimal. If $G^{CDX} \le 0$, then $\tau^*=t$ is optimal and $L^{CDX}_t=0$ a.s.
\end{proof}

 We observe that the drift function consists of two components. The first component in \eqref{G_n} accounts for the disagreement between investor and market on the fluctuation of market ex-dividend value, while the second integral term reflects the disagreement on the jumps of market's cumulative value arising from the   losses at default and the jumps in the  ex-dividend value.  Even though the market's cumulative value  $P^{CDX}$ in \eqref{PCDX} and the optimal expected liquidation value $V^{CDX}$ in \eqref{V_multi} are path-dependent, both the delayed liquidation premium $L^{CDX}$ in \eqref{L_multi}  and $G^{CDX}$ in \eqref{G_multi} depend only on  $t$ and $\lambda$ due to the special structure of ${C}^{CDX}$ given in \eqref{Cind22} .

To obtain the variational inequality of $L^{CDX}$, we recall that $\tilde{\lambda} = {\tilde{\mu}\lambda}/{\mu}$ and    the $\tilde{\Q}$-dynamics of default intensity $\lambda$:
\begin{align}
d\lambda_t=\tilde{\kappa}(\mu\tilde{\theta}-\lambda_t)dt+\sigma\sqrt{\mu\lambda_t}\,dW_t^{\tilde{\Q}}+\mu\eta \,d\Upsilon_t.\notag
\end{align}
The delayed liquidation premium $L^{CDX}(t,\lambda)$ as a function of time $t$ and $\Q$-default intensity $\lambda$ satisfies the  variational inequality
\begin{align}
\text{min}\bigg(&-\frac{\partial L^{CDX}}{\partial t}-\tilde{\kappa}(\mu\tilde{\theta}-\lambda)\frac{\partial L^{CDX}}{\partial \lambda}-\frac{\sigma^2\mu\lambda}{2} \frac{\partial L^{CDX}}{\partial
\lambda^2}+rL^{CDX}\notag\\
&-\frac{\tilde{\mu}\lambda}{\mu}\int_0^\infty \big(L^{CDX}(t,\lambda+\mu\eta z)-L^{CDX}(t,\lambda)\big) \tilde{m}(z)dz -G^{CDX}, \ L^{CDX}\bigg)=0, \label{L_index}
\end{align}
for  $(t,\lambda) \in [0,T) \times \mathbb{R}$, with   terminal condition $L(T,\lambda)=0$ for $ \lambda \in \mathbb{R}$.

We consider a numerical example for an index swap with constant losses at default. In this case, the integral term in \eqref{L_index} reduces to $L^{CDX}(t,\lambda+\mu\eta c)-L^{CDX}(t,\lambda)$, where $c$ is the constant loss. We employ the standard implicit PSOR iterative algorithm to solve $L^{CDX}$ by finite difference method with Neumann condition applied on the intensity boundary. There exist many   alternative numerical methods to solve variational inequality with an integral term (see, among others,  \cite{Andersen2000,Halluin2003}). We apply a second-order Taylor approximation to the difference $L^{CDX}(t,\lambda+\mu\eta c)-L^{CDX}(t,\lambda)\approx \partial_\lambda L^{CDX}(t,\lambda) \mu\eta c +\frac{1}{2}\partial_{\lambda\lambda}L^{CDX}(t,\lambda) (\mu\eta c)^2$.  In turn, these new partial derivatives are incorporated in the existing partial derivatives in \eqref{L_index}, rendering the variational inequality completely linear in $\lambda$, and thus, allowing for rapid computation.

We denote the investor's sell region $\mathcal{S}$ and delay region $\mathcal{D}$  by\begin{align}
\mathcal{S}^{CDX}&=\{(t,{\lambda})\in [0,T]\times \R : \ {L}^{CDX}(t,{\lambda})=0\}, \label{S_regioncdx}\\
\mathcal{D}^{CDX}&=\{(t,{\lambda})\in [0,T]\times \R : \ {L}^{CDX}(t,{\lambda})>0\}.\label{D_regioncdx}
\end{align}
On the other hand, we observe from \eqref{Cind22}  a one-to-one correspondence between the market's ex-dividend value ${C}^{CDX}$ of an index swap and its default intensity $\lambda$ for any fixed $t<T$, namely,
\begin{align}
\lambda=\frac{{C}^{CDX}-k_1(t,T) n-k_0(t,T)}{k_2(t,T)}.\label{ind_inverse}
\end{align}
Substituting (\ref{ind_inverse}) into (\ref{S_regioncdx}) and (\ref{D_regioncdx}), we can describe the sell region and delay region in terms of the observable market ex-dividend value ${C}^{CDX}$.

In Figure \ref{fig4}, we assume that the investor agrees with the market on all parameters except the speed of mean reversion for default intensity. In the case with  $\kappa=0.5<1=\tilde{\kappa}$ (Figure \ref{fig4} (left)), the investor's optimal liquidation strategy is to sell as soon as the market ex-dividend value of index swap $C^{CDX}$ reaches an upper boundary. In the case with   $\kappa=1>0.5=\tilde{\kappa}$ (Figure \ref{fig4} (right)), the sell region  is below the continuation region.

\begin{figure}[ht]
\centering
\begin{tabular}{cc}
\includegraphics[scale=0.52]{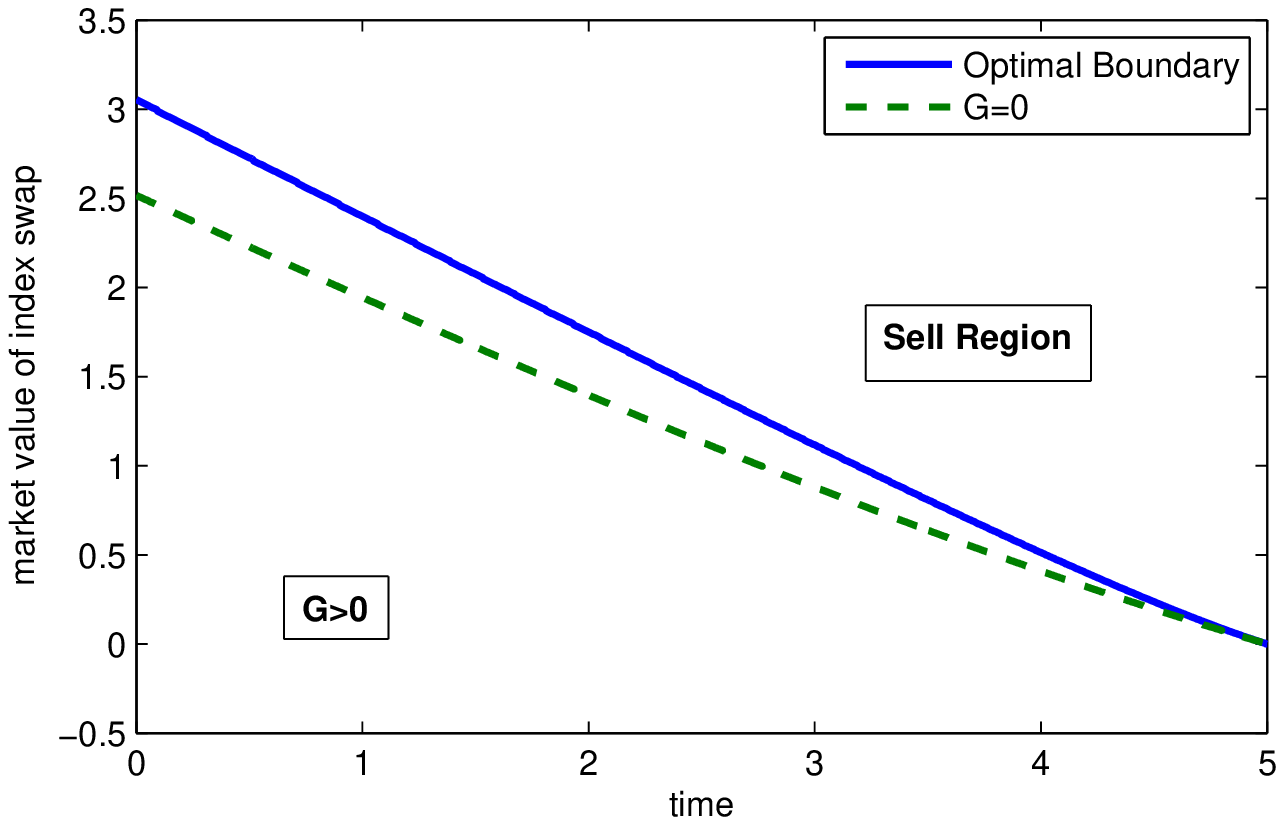} &
\includegraphics[scale=0.52]{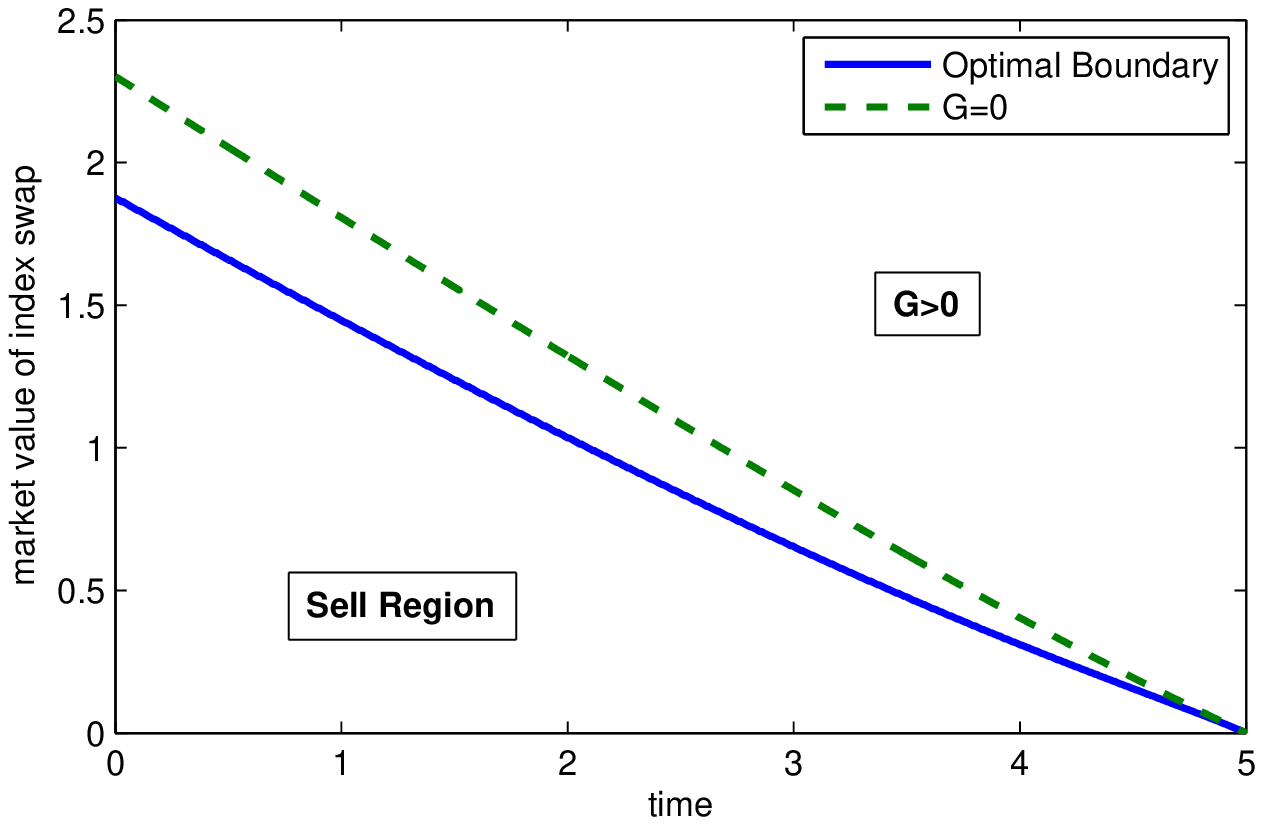} \\
\end{tabular}
\caption{\small{Optimal liquidation boundary in terms of market ex-dividend value of an index swap. We take $T=5$, $r=0.03$, $H=10$, $\eta= 0.25$, $\sigma=0.5$, $c=\tilde{c}=0.5$, $p^m_0=0.02$, $\theta=\tilde{\theta}=1$, and $\mu=\tilde{\mu}=1.1$.  \emph{Left panel}: When $\kappa=0.5<1=\tilde{\kappa}$, liquidation occurs at an upper boundary that decreases from $3$ to $0$ over $t\in[0,5]$. \emph{Right panel}: When $\kappa=1>0.5=\tilde{\kappa}$, the index swap is liquidated at a lower liquidation boundary, which  decreases from $1.9$ to $0$ over time. In both cases, the dashed line defined by $G^{CDX}=0$ lies within the continuation region. }}
\label{fig4}
\end{figure}

In summary, we have analyzed the optimal liquidation of a credit default index swap under a top-down credit risk model. The  selected model and contract specification give us tractable  analytical results that are amenable for numerical computation.  The top-down model implies  that the underlying credit portfolio can experience countably many defaults. As argued by Errais et al. \cite{Errais2010}, this feature is innocuous for a large diversified portfolio in practice since the likelihood of total default is negligible.

Our analysis here can be extended to the liquidation of  CDOs. Consider  a tranche with lower and higher attachment points ${K_1}, {K_2} \in[0,1]$ of  a CDO with $H$ names and  unit notionals. The tranche loss is a function of the accumulated loss $\Upsilon_t$, given by $\tilde{L}_t= (\Upsilon_t-{K_1}H)^+-(\Upsilon_t-{K_2}H)^+$, $t \in [0,T]$. With  premium rate $p_0^m$, the ex-dividend  market price of the CDO tranche for the protection buyer is
\begin{align}
\mathcal{C}^{CDO}(t, \lambda_t,\Upsilon_t)=\E^\Q\big\{ \int_{(t,T]} e^{-r(u-t)} \ d \tilde{L}_u\,|\mathcal{H}_t \big\}- \E^\Q\big\{p^m_0\int_{(t,T]} e^{-r(u-t)} (H({K_2}-{K_1})-\tilde{L}_u)\,du \,|\mathcal{H}_t \big\}\notag.
\end{align} Hence, the  CDO price is a function  of the accumulated loss $\Upsilon$, as opposed to  $N$   in the case of CDX (see Proposition \ref{ind_prop} above).

\section{Optimal Buying and Selling}Next, we adapt our model to study the  optimal buying and
selling problem. Consider an investor whose objective is to maximize the
revenue through a buy/sell transaction of a   defaultable claim $(Y, A, R,
\tau_d)$ with   market price process $P$   in (\ref{P_price}). The problem is
studied separately under two scenarios, namely, when  the short sale of the defaultable claim is permitted or prohibited. We shall analyze these problems under  the Markovian credit risk model in Section \ref{sect-credit derivative}.

If the investor seeks  to purchase a defaultable claim from the market, the optimal purchase timing problem and the associated \emph{delayed purchase premium} can be defined as:
\begin{align}
V_t^{b}=\essinf_{\tau^b \in \mathcal{T}_{t,T}}~\E^{\tilde{\Q}}\big\{e^{-\int_t^{\tau^b}  r_v  dv}P_{\tau^b}|\G_t\big\}, \quad \text{ and } \quad  L_t^{b}:=P_t-V_t^{b}\ge 0.\label{def_L_buy}
\end{align}

\subsection{Optimal Timing with Short Sale Possibility}
When short sale is permitted, there is no restriction on the
ordering of purchase time $\tau^b$ and  sale time $\tau^s$.  The  investor's investment timing is found from the optimal double-stopping problem:
\begin{align}
\mathcal{U}_t &:=\esssup_{\tau^b \in \mathcal{T}_{t,T},\tau^s \in \mathcal{T}_{t,T}}~\E^{\tilde{\Q}}\big\{e^{-\int_t^{\tau^s} r_v  dv}P_{\tau^s}-e^{-\int_t^{\tau^b} \! r_v  dv}P_{\tau^b}|\G_t\big\}.\nonumber
\end{align}
Since the defaultable claim will mature at  $T$, we interpret the choice of $\tau^b=T$ or  $\tau^s=T$ as no buy/sell transaction at $T$.

In fact, we can separate $\mathcal{U}$ into two optimal (single) stopping
problems. Precisely, we have
\begin{align}
\mathcal{U}_t &=\big(\esssup_{\tau^s \in \mathcal{T}_{t,T}}~\E^{\tilde{\Q}}\big\{e^{-\int_t^{\tau^s} r_vdv}P_{\tau^s}|\G_t\big\}-P_t\big)+\big(P_t-\essinf_{\tau^b \in \mathcal{T}_{t,T}}~\E^{\tilde{\Q}}\big\{e^{-\int_t^{\tau^b} r_v dv}P_{\tau^b}|\G_t\big\}\big)\notag\\
&=L_t+L_t^{b}.
\end{align}
Hence, we have separated $\mathcal{U}$ into a sum of  the \emph{delayed
liquidation premium} and   the \emph{delayed purchase premium}. As a result,  the optimal sale time $\tau^{s*}$ does not depend on the choice of the optimal purchase time $\tau^{b*}$.

The timing decision again
depends crucially on the sub/super-martingale properties of discounted
market price under measure $\tilde{\Q}$. Under the Markovian credit risk model in Section $3$, we can apply  Theorem
\ref{thm_main} to describe the optimal purchase and sale strategies in
terms of  the drift function $G(t,{\bf x})$ in (\ref{G_general}).

\begin{proposition}\label{prop_double1}
If $G(t,{\bf x})\geq0$ $\forall (t,{\bf x})\in [0,T]\times \R^n$, then it is
optimal to immediately buy  the defaultable claim   and hold it till maturity
$T$, i.e. $\tau^{b\ast}=t$ and $\tau^{s\ast}=T$ are optimal for $\mathcal{U}_t$. If $G(t,{\bf x})\leq0$ $\forall (t,{\bf x})\in [0,T]\times \R^n$, then it is optimal to immediately short sell
the claim  and maintain the position till  $T$, i.e. $\tau^{s\ast}=t$ and $\tau^{b\ast}=T$  are optimal for $\mathcal{U}_t$.
\end{proposition}

\subsection{Sequential Buying and Selling}
Prohibiting the  short sale of defaultable claims  implies the ordering: $\tau^b\leq \tau^s\leq T$.
Therefore, the investor's value function is
\begin{align}\label{U01}
U_t&:=\esssup_{\tau^b \in \mathcal{T}_{t,T},\tau^s \in
\mathcal{T}_{\tau^b,T}}\E^{\tilde{\Q}}\big\{e^{-\int_t^{\tau^s} \! r_v
dv}P_{\tau^s}-e^{-\int_t^{\tau^b} \! r_v
dv}P_{\tau^b}|\G_t\big\}.\end{align} The difference  $\mathcal{U}_t
-U_t\ge 0$ can be viewed as the cost of the short sale constraint to the investor.

As in Section \ref{sect-credit derivative}, we adopt the Markovian credit risk model, and derive from the $\tilde{\Q}$-dynamics of discounted market price in \eqref{P_SDE} to obtain \begin{align}
U_t&=\esssup_{\tau^b \in \mathcal{T}_{t,T},\tau^s \in
\mathcal{T}_{\tau^b,T}}~\E^{\tilde{\Q}}\big\{\int_{\tau^b}^{\tau^s}  (1-N_u) e^{-\int_t^u  r_v dv}G(u,{\bf X}_u) du |\G_t\big\}\nonumber\\
&=\1_{\{t<\tau_d\}}\esssup_{\tau^b \in \mathcal{T}_{t,T},\tau^s \in
\mathcal{T}_{\tau^b,T}}~\E^{\tilde{\Q}}\big\{\int_{\tau^b}^{\tau^s}  e^{-\int_t^u  (r_v+\tilde{\lambda}_v)  dv}G(u,{\bf X}_u) du |\F_t\big\}.
\end{align}
Using this probabilistic representation, we immediately deduce the optimal buy/sell strategy in the extreme cases analogues to Theorem \ref{thm_main}.
\begin{proposition}
If $G(t,{\bf x})\geq0$ $\forall (t,{\bf x})\in [0,T]\times \R^n$, then it is optimal to purchase the defaultable claim immediately and hold until maturity, i.e. $\tau^{b\ast}=t$ and $\tau^{s\ast}=T$ are optimal for $U_t$.

If $G(t,{\bf x})\leq0$ $\forall (t,{\bf x})\in [0,T]\times \R^n$, then it is optimal to never purchase the claim, i.e. $\tau^{b\ast}=\tau^{s\ast}=T$ is optimal for $U_t$.
\end{proposition}

Define $\hat{U}(t,{\bf X}_t)$  as the pre-default value of $U_t$, satisfying $U_t:=\1_{\{t<\tau_d\}}\hat{U}(t,{\bf X}_t)$.  We may view $\hat{U}(t,{\bf X}_t)$ as a \emph{sequential} optimal stopping problem.
\begin{proposition} \label{prop-seq_stop}
The value function $U_t$ in \eqref{U01} can be expressed in terms of the delayed liquidation premium $\hat{L}$ in \eqref{L_predef}. Precisely, we have
\begin{align}\hat{U}(t,{\bf X}_t)=
\esssup_{\tau^b \in \mathcal{T}_{t,T}}~\E^{\tilde{\Q}}\big\{e^{-\int_t^{\tau^b} (r_u+\tilde{\lambda}_u) du}\hat{L}_{\tau^b}|\F_t\big\}.\label{Uhat22}
\end{align}
\end{proposition}

\begin{proof}
We note that, after any purchase time $\tau^b$, the investor will face the  liquidation problem $V_{\tau^b}$ in \eqref{V_eq}.  Then using repeated conditioning, $U_t$ in \eqref{U01} satisfies
\begin{align}
U_t&= \esssup_{\tau^b \in \mathcal{T}_{t,T},\tau^s \in
\mathcal{T}_{\tau^b,T}}\E^{\tilde{\Q}}\big{\{}\big(e^{-\int_t^{\tau^b} r_u du}\E^{\tilde{\Q}}\big\{e^{-\int_{\tau^b}^{\tau^s} r_u du}P_{\tau^s}|\G_{\tau^b}\big\}-e^{-\int_t^{\tau^b}  r_u du}P_{\tau^b}\big)|\G_t\big{\}}\label{U22}\\
&\le \esssup_{\tau^b \in \mathcal{T}_{t,T}}~\E^{\tilde{\Q}}\big\{e^{-\int_t^{\tau^b}  r_u  du}(V_{\tau^b}-P_{\tau^b})|\G_t\big\}\label{Usec}\\
&= \esssup_{\tau^b \in \mathcal{T}_{t,T}}~\E^{\tilde{\Q}}\big\{e^{-\int_t^{\tau^b} r_u du}L_{\tau^b}|\G_t\big\}=\1_{\{t<\tau_d\}}\esssup_{\tau^b \in \mathcal{T}_{t,T}}~\E^{\tilde{\Q}}\big\{e^{-\int_t^{\tau^b} (r_u+\tilde{\lambda}_u) du}\hat{L}_{\tau^b}|\F_t\big\}.\label{seq_stop}
\end{align}
On the other hand, on the RHS of \eqref{Usec} we see that $V_{\tau^b} = \E^{\tilde{\Q}}\big\{e^{-\int_{\tau^b}^{\tau^{s*}} r_u du}P_{\tau^{s*}}|\G_{\tau^b}\big\}$, with  the optimal stopping time $\tau^{s*} := \inf\{  t \ge \tau^b : V_t  =  P_t\}$ (see \eqref{tau}). This is equivalent to  taking    the admissible stopping time  $\tau^{s*}$ for $U_t$ in \eqref{U01},  so the reverse of inequality \eqref{Usec} also  holds. Finally,  equating  \eqref{U22} and \eqref{seq_stop} and removing the default indicator, we arrive at \eqref{Uhat22}.
\end{proof}

According to Proposition \ref{prop-seq_stop}, the investor, who anticipates to liquidate the defautable claim after purchase, seeks to maximize the delayed liquidation premium when deciding to buy the derivative from the market. The practical implication of  representation \eqref{Uhat22} is that  we  first solve for the pre-default delayed liquidation premium $\hat{L}(t,{\bf x})$ by variational inequality \eqref{L_general_VI}. Then, using $\hat{L}(t,{\bf x})$ as input, we solve $\hat{U}$ by
\begin{align}
\text{min}\bigg(-\frac{\partial \hat{U}}{\partial t}(t,{\bf x})-\mathcal{L}_{\tilde{b},\tilde{\lambda}} \hat{U}(t,{\bf x}), \ \hat{U}(t,{\bf x})-\hat{L}(t,{\bf x})\bigg)=0, \quad (t,{\bf x}) \in [0,T) \times \mathbb{R}^n, \label{U_general_VI}
\end{align}
where $\mathcal{L}_{\tilde{b},\tilde{\lambda}}$ is defined in \eqref{oper}, and the  terminal condition is $\hat{U}(T,{\bf x})=0$, for $ {\bf x} \in \mathbb{R}^n$. In other words, the solution for  $\hat{L}(t,{\bf x})$ provides the investor's optimal liquidation boundary  after the purchase, and  the  variational inequality for $U(t,{\bf x})$ in \eqref{U_general_VI} gives the investor's optimal purchase boundary.

In Figure \ref{fig10}, we show a numerical example for a defaultable zero-coupon zero-recovery bond where interest rate is constant and $\lambda$ follows the CIR dynamics. The investor agrees with the market on all parameters except the speed of mean reversion for default intensity.  When $\kappa_\lambda<\tilde{\kappa}_\lambda$, the  optimal strategy is to buy as soon as the price enters the purchase region   and subsequently sell  at the (higher) optimal liquidation boundary. When $\kappa_\lambda>\tilde{\kappa}_\lambda$, the optimal liquidation boundary is below the purchase boundary. However,  it is possible that the investor buys at a lower price and subsequently sells at a higher price since both boundaries are increasing. It is also possible to buy-high-sell-low, realizing a  loss on these sample paths. On average, the optimal sequential buying and selling strategy enables the investor to profit from the price discrepancy. Finally,  when short sale is allowed, the investor's strategy follows the corresponding boundaries without the buy-first/sell-later constraint.

\begin{figure}[ht]
\centering
\begin{tabular}{cc}
\includegraphics[scale=0.458]{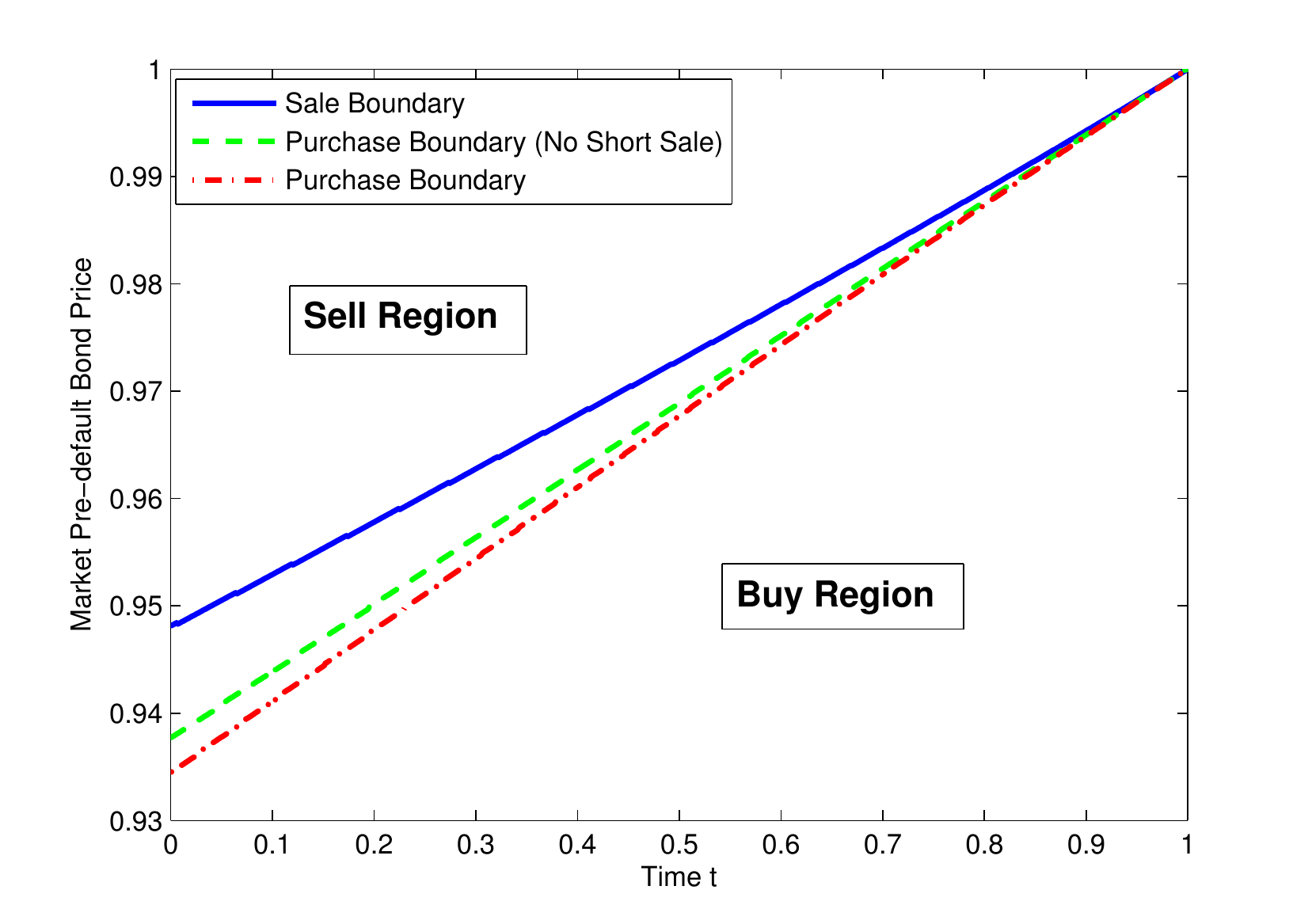} &
\includegraphics[scale=0.458]{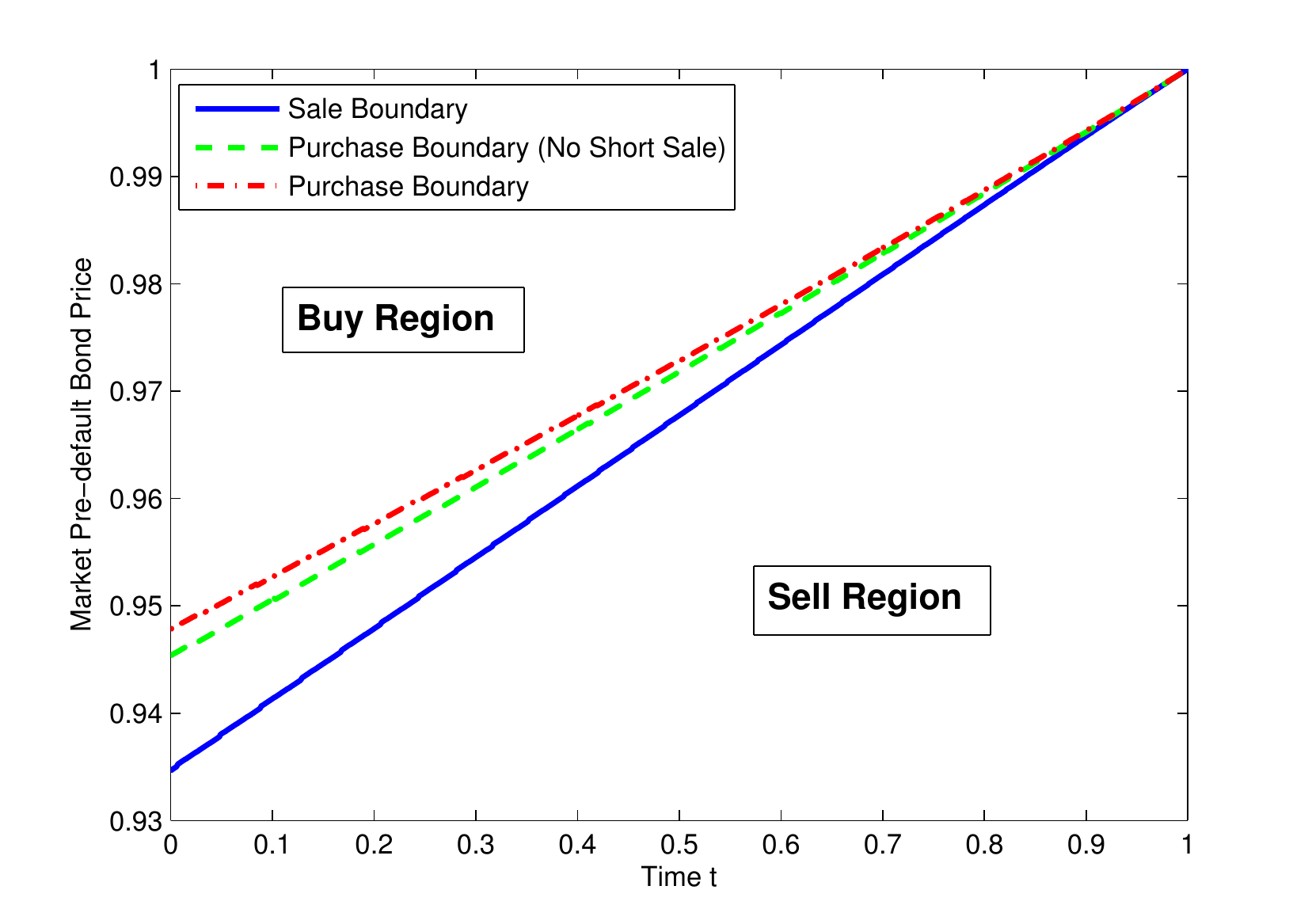} \\
\end{tabular}
\caption{\small{Optimal purchase and liquidation boundaries  in the CIR model. The common parameters are $T=1$, $r=0.03$, $\sigma=0.07$, $\mu=\tilde{\mu}=2$, and $\theta_\lambda=\tilde{\theta}_\lambda=0.015$. \emph{Left panel}: When $\kappa_\lambda=0.2<0.3=\tilde{\kappa}_\lambda$, the short sale constraint moves the  purchase boundary  higher. Both purchase boundaries, with or without short sale, are dominated by the liquidation boundary. \emph{Right panel}: When $\kappa_\lambda=0.3>0.2=\tilde{\kappa}_\lambda$, the short sale constraint moves the  purchase boundary  lower. The liquidation boundary lies below both purchase boundaries.}}
\label{fig10}
\end{figure}

\begin{remark}
In a related study, Leung and Ludkovski \cite{LeungLudkovski2011} also discuss  the  problem of sequential buying and selling  of  equity options without short sale possibility and with constant interest rate. In particular, the underlying stock admits a  \emph{local} default intensity modeled by $\hat{\lambda}(t,S_t)$, a deterministic function of time $t$ and  current stock price $S_t$.  In contrast,  our current model assumes  stochastic default intensity $\hat{\lambda}_t=\hat{\lambda}(t,{\bf X}_t)$ and interest rate $r(t,{\bf X}_t)$, driven by a stochastic factor vector ${\bf X}$. Hence, our optimal stopping value functions and buying/selling strategies depend on the stochastic factor ${\bf X}$, rather than the stock alone as in  \cite{LeungLudkovski2011}.\end{remark}

\section{Conclusions}
In summary, we have provided a flexible mathematical model for the
optimal liquidation of various credit derivatives under price discrepancy. We
have identified the situations where the optimal timing is trivial and also solved for the cases when sophisticated strategies are involved. The
optimal liquidation framework enables investors to quantify their views on
default risk, extract profit from price discrepancy, and perform more
effective risk management. Our model can also be modified and  extended
to incorporate single or multiple buying and selling decisions.

For future research, a natural direction is to consider credit derivatives trading under other default risk models. For multi-name credit derivatives, in contrast to the top-down approach taken in Section 5,  one  can consider the optimal liquidation problem under the bottom-up framework. Liquidation problems are also  important for derivatives portfolios  in general. To this end, the structure of dependency between multiple  risk factors is crucial in modeling  price dynamics.
 Moreover, it is both practically and mathematically interesting to allow for partial or sequential  liquidation (see e.g. \cite{Henderson2008, LeungYamazakiQF}). On the other hand,  market participants' pricing rules may vary due to different risk preferences. This leads to the  interesting question of how risk aversion influences their derivatives purchase/liquidation timing (see e.g. \cite{LeungLudkovski2} for the case of exponential utility).

\bibliographystyle{plain}
\singlespacing
\begin{small}
\bibliography{mybibLeungLiu2}
\end{small}

\end{document}